\documentclass[11pt]{article}

\usepackage[utf8]{inputenc}
\usepackage[margin=1in]{geometry}
\usepackage{subcaption}

\usepackage[numbers]{natbib}

\usepackage{amsfonts, amsthm, amsmath, amssymb}
\usepackage{mathtools}
\usepackage{xifthen}
\usepackage{xcolor}
\usepackage{tikz}
\usetikzlibrary{backgrounds}
\usetikzlibrary{arrows}
\usetikzlibrary{shapes,shapes.geometric,shapes.misc}

\tikzstyle{tikzfig}=[baseline=-0.25em,scale=0.5]

\pgfkeys{/tikz/tikzit fill/.initial=0}
\pgfkeys{/tikz/tikzit draw/.initial=0}
\pgfkeys{/tikz/tikzit shape/.initial=0}
\pgfkeys{/tikz/tikzit category/.initial=0}

\pgfdeclarelayer{edgelayer}
\pgfdeclarelayer{nodelayer}
\pgfsetlayers{background,edgelayer,nodelayer,main}

\tikzstyle{none}=[inner sep=0mm]

\tikzstyle{every loop}=[]

\tikzstyle{nodo}=[fill=black, draw=black, shape=circle, inner sep=1.5pt]
\tikzstyle{class 1}=[fill=red, draw=red, shape=circle, inner sep=1.5pt]
\tikzstyle{class 2}=[fill=blue, draw=blue, shape=circle, inner sep=1.5pt]
\tikzstyle{class 3}=[fill=green, draw=green, shape=circle, inner sep=1.5pt]
\tikzstyle{class 4}=[fill=magenta, draw=magenta, shape=circle, inner sep=1.5pt]
\tikzstyle{class 5}=[fill={rgb,255: red,255; green,128; blue,0}, draw={rgb,255: red,255; green,128; blue,0}, shape=circle, inner sep=1.5pt]
\tikzstyle{class 6}=[fill=yellow, draw=yellow, shape=circle, inner sep=1.5pt]

\tikzstyle{arista}=[-, fill=none, draw=black]

\usepackage{booktabs}
\usepackage{graphicx}
\usepackage{xspace}
\usepackage{comment}
\usepackage{etoolbox}
\usepackage{relsize}
\usepackage[title]{appendix}

\usepackage{hyperref, varioref}
\usepackage[capitalize,noabbrev]{cleveref}

\theoremstyle{plain}
\newtheorem{theorem}{Theorem}
\newtheorem{lemma}[theorem]{Lemma}
\newtheorem{corollary}[theorem]{Corollary}

\newtheorem{claim}[theorem]{Claim}

\theoremstyle{definition}
\newtheorem{definition}[theorem]{Definition}

\theoremstyle{remark}

\newcommand{\contract}[2]{#1|_{#2}}
\renewcommand{\O}{\mathcal{O}}
\newcommand{\N}{\mathbb{N}}
\newcommand{\declarefunctionOld}[2][]{%
    \ifthenelse{\isempty{#1}}{%
        \expandafter\newcommand\csname #2\endcsname{{\small\textsf{#2}}}%
    }{%
        \expandafter\newcommand\csname #1\endcsname{{\small\textsf{#2}}}%
    }%
}

\newcommand{\declarefunction}[2][]{%
    \ifthenelse{\isempty{#1}}{%
        \expandafter\newcommand\csname #2\endcsname{{\textsf{\textup{#2}}}\xspace}%
    }{%
        \expandafter\newcommand\csname #1\endcsname{{\textsf{\textup{#2}}}\xspace}%
    }%
}

\declarefunction{thin}
\declarefunction{lmimw}
\declarefunction{tw}
\declarefunction{pw}
\declarefunction{mw}
\declarefunction{s}
\declarefunction{nd}
\declarefunction{tc}
\declarefunction{vs}
\declarefunction{vc}
\declarefunction{p}
\declarefunction{si}
\declarefunction{cc}
\declarefunction{mc}
\declarefunction{cluster}
\declarefunction{interval}
\newcommand{\G}{\mathcal{G}}
\newcommand{\Gmc}{\text{$\G$-$\mc$}}

\newcommand{\intervalmc}{\text{$\interval$-$\mc$}}

\newcommand{\clustermc}{\text{$\cluster$-$\mc$}}
\newcommand{\cm}{\clustermc}

\newcommand{\FPT}{{\smaller[0.83]\textsf{\textup{FPT}}}\xspace}
\newcommand{\XP}{{\smaller[0.83]\textsf{\textup{XP}}}\xspace}
\newcommand{\NP}{{\smaller[0.83]\textsf{\textup{NP}}}}
\newcommand{\coNP}{{\smaller[0.83]\textsf{\textup{coNP}}}}
\newcommand{\poly}{{\smaller[0.83]\textsf{\textup{poly}}}}

\newcommand{\Hip}{\NP~$\not\subseteq$~\coNP/\poly}
\newcommand{\notHip}{\NP~$\subseteq$~\coNP/\poly}

\newcommand{\defproblem}[3]{
 \vspace{3mm}
\noindent\fbox{
 \begin{minipage}{0.96\textwidth}
 \begin{tabular*}{\textwidth}{@{\extracolsep{\fill}}lr} \textsc{#1} & \\ \end{tabular*}
 {\bf{Input:}} #2 \\
 {\bf{Question:}} #3
 \end{minipage}
 }
 \vspace{3mm}
}

\usepackage{color}
\DeclarePairedDelimiter\abs{\lvert}{\rvert}

\usetikzlibrary{fit}
\usetikzlibrary{shapes,shapes.geometric,arrows,fit,calc,positioning,automata,arrows}
\usetikzlibrary{patterns,hobby}
\usepackage{pgfplots}
\pgfplotsset{compat=1.15}

\tikzstyle{line}=[draw]

\tikzset{ n1/.style={circle,scale=1.0},
n2/.style={circle,fill=black,scale=0.5},
n3/.style={circle,draw,fill=black,draw=black,text=white,scale=0.9},
e1/.style={line width=0.1mm}, e3/.style={draw=black,line
width=0.7mm}, inter/.style={line width=0.85mm},
e2/.style={draw=black,line width=0.5mm}, c1/.style={line
width=0.3mm}, c2/.style={line width=0.2mm} }

\newcommand{\vertex}[4][black]{
    \draw[#1, fill=#1, inner sep=0pt] (#2, #3) circle (0.12) node(#4){};
}

\newcommand{\minivertex}[4][black]{
    \draw[#1, fill=#1, inner sep=0pt] (#2, #3) circle (0.1) node(#4){};
}

\newcommand{\vertexLabel}[3][above]{
    \path (#2) node[#1]{#3};
}

\title{Computing parameters that generalize interval graphs using restricted modular partitions}

\author{
  Flavia Bonomo-Braberman\thanks{Partially supported by ANPCyT (PICT-2021-I-A-00755). ORCID: 0000-0002-9872-7528.} \\
  \small Departamento de Computaci\'{o}n, FCEyN, Universidad de Buenos Aires, Argentina \\
  \small Instituto de Investigaci\'{o}n en Ciencias de la Computaci\'{o}n (ICC), CONICET-UBA, Argentina
  \and
  Eric Brandwein\thanks{Partially supported by IRP SINFIN (France-Argentina), CONICET (PIP 11220200100084CO), and UBACyT (20020220300079BA). ORCID: 0009-0003-2559-7173.} \\
  \small Departamento de Computaci\'{o}n, FCEyN, Universidad de Buenos Aires, Argentina \\
  \small Instituto de Investigaci\'{o}n en Ciencias de la Computaci\'{o}n (ICC), CONICET-UBA, Argentina
  \and
  Ignasi Sau\thanks{Supported by French project ELIT (ANR-20-CE48-0008-01). ORCID: 0000-0002-8981-9287.} \\
  \small LIRMM, Universit\'{e} de Montpellier, CNRS, Montpellier, France
}

\date{}
\begin{document}

\maketitle

\begin{abstract}
Recently, Lafond and Luo [MFCS 2023] defined the $\mathcal{G}$-modular cardinality of a graph $G$ as the minimum size of a partition of $V(G)$ into modules that belong to a graph class $\mathcal{G}$. We analyze the complexity of calculating parameters that generalize interval graphs when parameterized by the $\mathcal{G}$-modular cardinality, where $\mathcal{G}$ corresponds either to the class of interval graphs or to the union of complete graphs. Namely, we analyze the complexity of computing the thinness and the simultaneous interval number of a graph.

We present a linear kernel for the \textsc{Thinness} problem parameterized by the $\textsf{interval}$-modular cardinality and an \textsf{FPT} algorithm for \textsc{Simultaneous Interval Number} when parameterized by the $\textsf{cluster}$-modular cardinality plus the solution size. The $\textsf{interval}$-modular cardinality of a graph is not greater than the $\textsf{cluster}$-modular cardinality, which in turn generalizes the neighborhood diversity and the twin-cover number. Thus, our results imply a linear kernel for \textsc{Thinness} when parameterized by the neighborhood diversity of the input graph, \textsf{FPT} algorithms for \textsc{Thinness} when parameterized by the twin-cover number and vertex cover number, and \textsf{FPT} algorithms for \textsc{Simultaneous Interval Number} when parameterized by the neighborhood diversity plus the solution size, twin-cover number, and vertex cover number. To the best of our knowledge, prior to our work no parameterized algorithms (\textsf{FPT} or \textsf{XP}) for computing the thinness or the simultaneous interval number were known.

On the negative side, we observe that \textsc{Thinness} and \textsc{Simultaneous Interval Number} parameterized by tree\-width, pathwidth, bandwidth, (linear) mim-width, clique-width, modular-width, or even the thinness or simultaneous interval number themselves, admit no polynomial kernels assuming $\mathsf{NP} \not\subseteq \mathsf{coNP}/\mathsf{poly}$.
\end{abstract}


\section{Introduction}
\label{sec:intro} 
The class of interval graphs~\citep{interval-graphs} has many useful properties that makes it amenable to fast algorithms for many common problems~\citep{Ola-interval,R-PR-interval}. Thus, many graph parameters were defined aiming at quantifying the distance of a graph to an interval graph in some way, including the \emph{interval number}~\citep{interval-number}, the \emph{track number}~\citep{multitrack-interval-graphs}, the \emph{thinness}~\citep{M-O-R-C-thinness}, and, more recently, the \emph{simultaneous interval number}~\citep{Milanic-sim}, among others. If one of these parameters is small for a class of graphs $\G$, some properties of interval graphs are possibly maintained on $\G$, and therefore some problems that are easy on interval graphs might be easy on $\G$ as well. In this work we focus on two of these parameters: the thinness and the simultaneous interval number. In particular, we will be interested in computing them efficiently. We defer their formal definitions to \cref{sec:preliminaries}. We now proceed to briefly motivate the study of these graph parameters.

\paragraph{Thinness.} The thinness of a graph was introduced by \citet{M-O-R-C-thinness} to solve the \textsc{Maximum Independent Set} problem in polynomial time when a suitable representation of the input graph is given. The concept of thinness can be seen as a generalization of the characterization of interval graphs as those that admit an ordering $<$ of the vertices such that for every triple of vertices $u < v < w$, if $(u,w)$ is an edge, then $(v,w)$ is also an edge~\citep{forbidden-ordered-subgraphs}. The thinness of a graph is the minimum size of a partition of the vertices such that every triple of vertices $u < v < w$ either meets the condition mentioned above, or $u$ and $v$ belong to different parts of the partition. Observe that interval graphs are exactly the graphs of thinness one.

Given a suitable representation for a graph with thinness bounded by $t$, many other problems including \textsc{Clique} were shown to be solvable in \XP-time parameterized by $t$, while problems like \textsc{Capacitated $k$-Coloring} and \textsc{List $k$-Coloring} are \XP parameterized by $t$ plus the number of color classes~\citep{M-O-R-C-thinness,B-M-O-thin-tcs,C-T-V-Z--H-top-thin,tesis-diego}. It is worth noting that \textsc{Chromatic Number} is \NP-complete for $t=2$, even if the representation is given~\citep{BBOSSS-thin-coloring}.

\paragraph{Simultaneous interval number.} The simultaneous interval number was recently introduced by \citet{Milanic-sim}. It generalizes the class of interval graphs by assigning not only an interval to each vertex, but also a subset of \emph{labels} from the set $\{1,\dots,d\}$, with $d\in \N$. The representation must be such that two vertices are adjacent in the graph if and only if their intervals \textbf{and} their sets of labels intersect. The simultaneous interval number is the minimum number $d$ for which there exists such a representation.

In their work, \citet{Milanic-sim} show that when given a suitable representation for a graph with simultaneous interval number bounded by $d$, \textsc{Clique} is \FPT parameterized by $d$, and \textsc{Independent Set} and \textsc{Dominating Set} are \FPT parameterized by $d$ plus the solution size. On the other hand, \textsc{Chromatic Number} is \NP-complete for $d=2$, even if the representation is given.

Unfortunately, determining the value of these two parameters is \NP-hard on general graphs \citep{Milanic-sim, thinness-np-complete}. This prompts the investigation of algorithms for these two parameters on restricted classes of inputs. For a start, \textsc{Thinness} was solved in polynomial time for cographs by \citet{tesis-diego} and for trees by \citet{thinness-of-trees}. Some classes of graphs with bounded thinness, where a bounded-size consistent solution can be obtained in polynomial time, were identified by \citet{B-B-M-P-convex-jcss}. On the other hand, not much is known about the complexity of \textsc{Simultaneous Interval Number} on restricted graph classes. 

\paragraph{Our contribution.} Given the inherent hardness discussed above, a natural approach is to search for parameterized algorithms for the \textsc{Thinness} and \textsc{Simultaneous Interval Number} problems, which ask to find the value of the corresponding parameters for a given graph. In particular, previous to this work there were no known \FPT or even \XP algorithms for computing any of the two parameters.

The field of parameterized complexity focuses on finding fast algorithms for problems when a natural number that is part of the input, the \emph{parameter}, is small. In particular, one focus is on finding \XP algorithms, which are algorithms that take time $f(k) \cdot n^{g(k)}$ for some computable functions $f$ and $g$, where $k$ is the parameter and $n$ is the size of the input. Another focus is on \FPT algorithms, which instead take time $f(k) \cdot n^c$ for some constant $c$. One way of finding an \FPT algorithm for a problem is to find a \emph{kernel} for it. A \emph{kernelization} (or \emph{kernel}) for a problem $Q$ is a polynomial-time algorithm that transforms each input $(x, k)$ into an \emph{equivalent} instance $(x', k')$ such that $\abs{x'}, k' \leq g(k)$ for some computable function $g$. We say that $g(k)$ is the \emph{size} of the kernel. Finding a kernel for a problem is equivalent to proving that it is \FPT~\citep[Lemma 2]{parameterized-algorithms}. Kernels of polynomial size are of particular interest due to their practical applications.

The existence of \FPT or even \XP algorithms for \textsc{Thinness} or \textsc{Simultaneous Interval Number} parameterized by the solution size (the numbers $t$ and $d$, respectively) is open. These would be the most natural parameterizations to consider, but unfortunately, these algorithms seem difficult to obtain. We thus focus instead on parameterizing the problems by other relatively natural parameters based on partitions of the input graph into modules.

The concept of \emph{$\G$-modular cardinality} of a graph $G$ for a fixed graph class $\G$ was introduced by \citet{modular-partitions-MFCS23} as the minimum number of parts in a partition of $G$ into modules, each inducing a graph in the class $\G$. In their article, Lafond and Luo study the complexity of domination problems parameterized by the $\G$-modular cardinality for some classes $\G$, and they use their results to show that these problems are $W[1]$-hard when parameterized by the clique-width, the modular-width, and other parameters. Additionally, positive results for problems parameterized by the $\G$-modular cardinality imply positive results when parameterized by the neighborhood diversity, the twin-cover number, and the vertex cover number. We will be focusing on the parameterizations by the $\cluster$-modular cardinality and the $\interval$-modular cardinality, where $\cluster$ is the class of unions of complete graphs, and $\interval$ is the class of interval graphs. 

If a parameter $p(G)$ is not greater than a function $f$ of another parameter $s(G)$ for every graph $G$, we say that $p$ \emph{generalizes} $s$. In particular, the $\interval$-modular cardinality generalizes the $\cluster$-modular cardinality. This means that obtaining a kernel for a problem parameterized by the $\interval$-modular cardinality of the input graph automatically yields a kernel for the problem parameterized by the $\cluster$-modular cardinality. In turn, the $\cluster$-modular cardinality generalizes the neighborhood diversity and the twin-cover number (both of which generalize the vertex cover number), and is less general than the modular-width. A relationship diagram of thinness and other graph width parameters can be found in a work by \citet{B-B-M-P-convex-jcss}. We combine in \cref{fig:diagram} the figure by \citet{B-B-M-P-convex-jcss} with figures by \citet{Milanic-sim}, \citet{Belm-mw}, and \citet{modular-width}, and further include the $\cluster$- and $\interval$-modular cardinalities.

\begin{figure}[ht]
    \begin{center}
        \resizebox{0.8\textwidth}{!}{%
    \begin{tikzpicture}[yscale=.7,xscale=1.1]
    \node[draw,rounded corners=3,inner sep=2.6pt] at (0,8.5) (mmw) {\phantom{j}mim-width\phantom{j}};
    \node[draw,rounded corners=3,inner sep=2.6pt] at (-3,7.5) (lmmw) {\phantom{j}linear mim-width\phantom{j}};
    \node[draw,rounded corners=3,inner sep=2.6pt] at (0,7) (cw) {\phantom{j}clique-width\phantom{j}};
    \node[draw,rounded corners=3,inner sep=2.6pt] at (0,5.5) (tw) {\phantom{j}treewidth\phantom{j}};
    \node[draw,rounded corners=3,inner sep=2.6pt] at (0,2.5) (pw) {\phantom{j}pathwidth\phantom{j}};
    \node[draw,rounded corners=3,inner sep=2.6pt] at (-9.75,1) (nd) {\phantom{j}neighborhood diversity\phantom{j}};
    \node[draw,rounded corners=3,inner sep=2.6pt] at (-5.75,1) (tc) {\phantom{j}twin-cover number\phantom{j}};
    \node[draw,rounded corners=3,inner sep=2.6pt] at (0,1) (bw) {\phantom{j}bandwidth\phantom{j}};
    \node[draw,rounded corners=3,inner sep=2.6pt] at (-7.75,-1) (vc) {\phantom{j}vertex cover number\phantom{j}};
    \node[draw,rounded corners=3,inner sep=2.6pt] at (-7.75,4.5) (im) {\phantom{j}$\interval$-modular cardinality\phantom{j}};
    \node[draw,rounded corners=3,inner sep=2.6pt] at (-7.75,2.5) (cm) {\phantom{j}$\cluster$-modular cardinality\phantom{j}};
    \node[draw,rounded corners=3,inner sep=2.6pt] at (-11,6) (mw) {\phantom{j}modular-width\phantom{j}};
    \node[draw,rounded corners=3,inner sep=2.6pt] at (-3,5.5) (th) {\phantom{j}thinness\phantom{j}};
    \node[draw,rounded corners=3,inner sep=2.6pt] at (-3,3.5) (si) {\phantom{j}simultaneous interval number\phantom{j}};
    \draw [->, shorten >= 1pt] (mmw) -- (cw);
    \draw [->, shorten >= 1pt] (cw) -- (tw);
    \draw [->, shorten >= 1pt] (tw) -- (pw);
    \draw [->, shorten >= 1pt] (pw) -- (bw);
    \draw [->, shorten >= 4pt] (mmw) -- (lmmw);
    \draw [->, shorten >= 1pt] (lmmw) -- (th);
    \draw [->, shorten >= 2pt] (si) -- (pw);
    \draw [->, shorten >= 1pt] (pw) -- (vc);
    \draw [->, shorten >= 3pt] (nd) -- (vc);
    \draw [->, shorten >= 3pt] (tc) -- (vc);
    \draw [->, shorten >= 3pt] (cm) -- (tc);
    \draw [->, shorten >= 3pt] (cm) -- (nd);
    \draw [->, shorten >= 1pt] (mw) to[bend right] (cm);
    \draw [->, shorten >= 1pt] (cw) -- (mw);
    \draw [->, shorten >= 1pt] (th) -- (si);
    \draw [->, shorten >= 7pt] (th) -- (im);
    \draw [->] (im) -- (cm);
    \draw [->, shorten >= 3pt] (si) -- (tc);
    \end{tikzpicture}}
    \end{center}
    \caption{The relationships between the different width parameters that we consider in this work, and some classical ones. Here, an arrow from parameter $A$ to parameter $B$ means that $A$ is bounded by a function on $B$. Most of the relationships were known and appear in diagrams of~\citep{Milanic-sim,Belm-mw,B-B-M-P-convex-jcss,modular-width}. The remaining ones are proved in \cref{prop:interval-modular-cardinality-parametrization}, \cref{prop:neighborhood-partition-is-cluster-modular-partition,lemma:twin-cover-hcupn,lemma:si-tc}, and the results in \cref{sec:modular-cardinality,sec:parameterizations}.
    {To explain the incomparability between thinness and modular-width, between simultaneous interval number and both $\interval$-modular cardinality and neighborhood diversity, and between $\interval$-modular cardinality and modular-width, observe that cographs have modular-width~$2$ and unbounded thinness~\citep{tesis-diego}, and complete bipartite graphs have neighborhood diversity~$2$ and unbounded simultaneous interval number~\citep{Milanic-sim}, while paths have simultaneous interval number~$1$, $\interval$-modular cardinality~$1$, and unbounded modular-width.} Cographs have bounded linear mim-width, but inspired by the ideas of \citet{H-A-R-lin-mim-trees}, we can also build graph families of bounded modular-width and unbounded linear mim-width (see \cref{sec:modw-vs-lmimw}). 
    }\label{fig:diagram}
\end{figure}
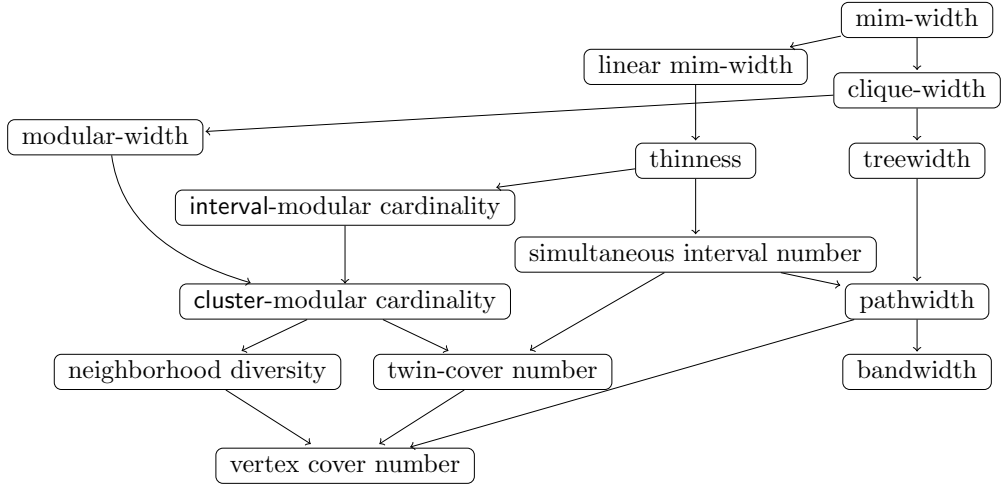

The article that introduces the $\G$-modular cardinality~\citep{modular-partitions-MFCS23}\footnote{Full proofs can be found in~\citep{modular-partitions-arxiv}.} presents a polynomial-time algorithm to obtain it for an arbitrary graph when $\G$ meets conditions that are satisfied by the classes $\cluster$ and $\interval$, but their algorithm does not explicitly run in linear time. Thus, we first show that computing the $\cluster$- and $\interval$-modular cardinality of a graph can be done in linear time in \cref{prop:calculate-cluster-module-partition,prop:calculate-interval-module-partition}. 
 In the case of the $\interval$-modular cardinality, this necessitates the results in \cref{prop:contract-complete-module,prop:contract-interval-module}, which allow the elimination of most of the vertices of a module that induces an interval graph without modifying the thinness of the whole graph. These results cannot be easily generalized to modules that induce graphs with bounded thinness; we show this in \cref{sec:counterexample} by providing a counterexample to a possible generalization of them. Then, we use these results to present linear kernels for the \textsc{Thinness} problem when parameterized by the $\interval$-modular cardinality (\cref{prop:interval-modular-cardinality-parametrization}) and neighborhood diversity (\cref{prop:neighborhood-diversity-parametrization}) of the input graph, and exponential kernels when parameterized by the twin-cover number (\cref{prop:twin-cover-parametrization}) and vertex cover number (\cref{prop:vertex-cover-parametrization}). Additionally, we show \FPT algorithms for the \textsc{Simultaneous Interval Number} problem parameterized by the $\cluster$-modular cardinality plus the solution size (\cref{theorem:si-param-cm}), neighborhood diversity plus the solution size (\cref{corol:si-param-nd-vc}), twin-cover number (\cref{cor:si-param-twin-cover}), and vertex cover number (\cref{corol:si-param-nd-vc}). For kernelization lower bounds, we observe in \cref{theorem:thinness-kernelization-lower-bound} that no polynomial kernels are likely to exist for the \textsc{Thinness} or \textsc{Simultaneous Interval Number} problems when parameterized by a number of graph parameters, including treewidth, modular-width, and the thinness and simultaneous interval number themselves. Finally, in \cref{sec:modw-vs-lmimw} we complete the relationships of \cref{fig:diagram} by proving that modular-width and linear mim-width do not generalize one another.

In order to achieve the results mentioned above, we highlight, in particular, the following technical contributions: determining in \cref{prop:contract-interval-module} that almost all vertices belonging to a module that induces an interval graph can be eliminated without modifying the thinness of the graph, and presenting an \FPT algorithm for \textsc{Simultaneous Interval Number} when parameterized by the $\cluster$-modular cardinality in \cref{theorem:si-param-cm} by reducing the problem to a limited number of small instances.

\paragraph{Organization of the paper.} In \cref{sec:preliminaries} we state some necessary definitions. In \cref{sec:reducing-modules} we prove some lemmas about the thinness of a graph after removing vertices that belong to the same module. In \cref{sec:modular-cardinality} we show linear-time algorithms to compute the $\cluster$- and $\interval$-modular cardinality of a graph. In \cref{sec:parameterizations} we use the results from \cref{sec:reducing-modules} to compute kernels for the {\sc Thinness} problem parameterized by the $\interval$- and $\cluster$-modular cardinality, neighborhood diversity, twin-cover number, and vertex cover number of the input graph.
In \cref{sec:sim-int} we present \FPT algorithms for the {\sc Simultaneous Interval Number} problem parameterized by the $\cluster$-modular cardinality plus the solution size, neighborhood diversity plus the solution size, twin-cover number, and vertex cover number. We summarize our algorithmic results and present some open problems in \cref{sec:open-problems}.
In \cref{sec:modw-vs-lmimw} we show that there exist graph families with bounded modular-width and unbounded linear mim-width, completing the diagram of \cref{fig:diagram}. In \cref{sec:counterexample} we show a counterexample to a possible generalization of the results of \cref{sec:reducing-modules}. In \cref{sec:lower-bounds} we observe that polynomial kernels are unlikely to exist for the \textsc{Thinness} {and \textsc{Simultaneous Interval Number} problems} parameterized by some common graph parameters, including treewidth, modular-width, and the thinness {and simultaneous interval number themselves}; {we moved this section to the appendix because its proofs are nowadays standard in the field of parameterized complexity.}

\section{Preliminaries}
\label{sec:preliminaries}
\paragraph{Graphs.} All graphs in this work are finite, undirected, and have no loops or multiple edges. 
Let $G$ be a graph. We denote by $V(G)$ and $E(G)$ its vertex and edge set, respectively. We denote by $N(v)$ and $N[v]$, respectively,
the neighborhood and closed neighborhood of a vertex $v \in V(G)$. The complement of $G$ is denoted $\overline{G}$.
Let $X \subseteq V(G)$.
We denote by $G[X]$ the subgraph of $G$ induced by $X$, and by $G \setminus X$ the graph $G[V(G) \setminus X]$. If $X = \{v\}$, we simply write $G \setminus v$.

A subset $X$ of $V(G)$ is \emph{complete} (or a \emph{clique}) if every pair of vertices in $X$ are adjacent, and \emph{independent} if no pair of vertices in $X$ are adjacent. A \emph{cluster} is a union of cliques. We denote the class of all cluster graphs by $\cluster$, and the class of all interval graphs by $\interval$.

\paragraph{Modules.} A \emph{module} of a graph $G$ is a subset $X$ of its vertices such that every vertex not in $X$ is either adjacent to all vertices in $X$, or to none of them. We say that a vertex is a \emph{neighbor} of or is \emph{adjacent} to the module if it is not in the module and is adjacent to the vertices in the module.
A \emph{modular partition} of a graph is a partition of its vertex set into modules. The \textit{contraction of $X$ into a  vertex} is defined as the graph $\contract{G}{X}$ resulting of removing all vertices in $X$ from $G$ and adding a new vertex $x$ adjacent to all neighbors of a vertex in $X$.

A module is \emph{proper} if it is not $V(G)$. 
A proper module $X$ is \emph{maximal} if there is no strict superset of $X$ that is a proper module. The maximal modules of a graph $G$ form a partitive family~\citep{modular-decomposition-partitive}, and thus a unique decomposition tree can be defined over a graph $G$ such that each node is a maximal module that does not overlap with other maximal modules (see \cref{fig:mod-decomp}). The \emph{modular decomposition tree} of $G$ is a rooted tree $T$ defined recursively as follows:
\begin{itemize}
    \item The root of the tree is the entire graph $G$.
    \item If $G$ contains only one vertex, then the root is called a \emph{Leaf} node.
    \item If $G$ is disconnected, then the children of the root are the modular decomposition trees of the connected components of $G$, and the root is called a \emph{Parallel} node.
    \item If the complement of $G$ is disconnected, then the children of the root are the modular decomposition trees of the connected components of $\overline{G}$, and the root is a \emph{Series} node.
    \item Otherwise, the root is a \emph{Prime} node, and its children are the modular decomposition trees of the maximal modules of $G$.
\end{itemize}

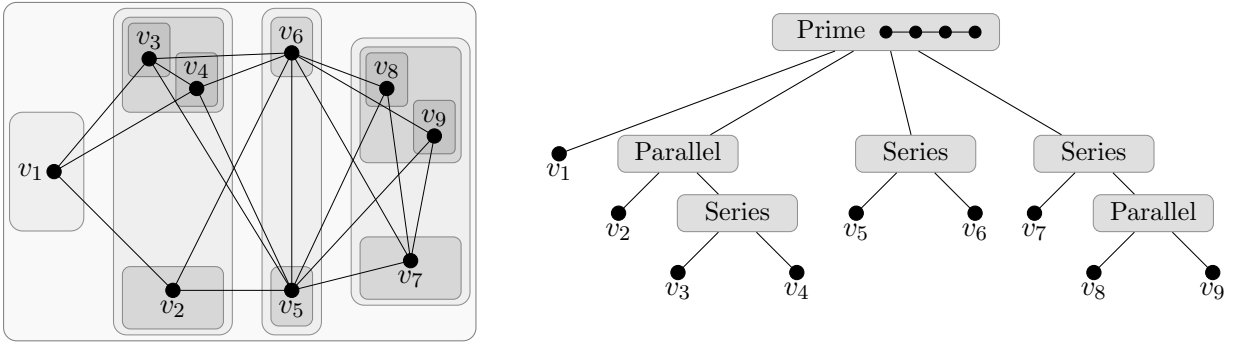
\begin{figure}[ht]
    \centering
    \resizebox{\textwidth}{!}{%
\begin{tikzpicture}[scale=0.8]

\draw[gray,fill=gray!5,rounded corners=5]
     (-0.85,-0.85) rectangle (7.1,4.85);

\draw[gray!30!black,fill=gray!20,rounded corners=5,opacity=0.5]
     (-0.75,1) rectangle (0.5,3);

\draw[gray!30!black,fill=gray!20,rounded corners=5,opacity=0.5]
     (1,-0.75) rectangle (3,4.75);

\draw[gray!30!black,fill=gray!20,rounded corners=5,opacity=0.5]
     (3.5,-0.75) rectangle (4.5,4.75);

\draw[gray!30!black,fill=gray!20,rounded corners=5,opacity=0.5]
     (5,-0.25) rectangle (7,4.25);

\draw[gray!30!black,fill=gray!50,rounded corners=3,opacity=0.5]
     (1.15,3) rectangle (2.85,4.6);

\draw[gray!30!black,fill=gray!50,rounded corners=3,opacity=0.5]
     (1.15,-0.65) rectangle (2.85,0.4);

\draw[gray!30!black,fill=gray!50,rounded corners=3,opacity=0.5]
     (5.15,2.15) rectangle (6.85,4.1);

\draw[gray!30!black,fill=gray!50,rounded corners=3,opacity=0.5]
     (5.15,-0.15) rectangle (6.85,0.9);

\draw[gray!30!black,fill=gray!50,rounded corners=3,opacity=0.5]
     (3.65,3.6) rectangle (4.35,4.6);

\draw[gray!30!black,fill=gray!50,rounded corners=3,opacity=0.5]
     (3.65,-0.6) rectangle (4.35,0.4);

\draw[gray!30!black,fill=gray!60,rounded corners=2,opacity=0.5]
     (1.25,3.6) rectangle (1.95,4.5);

\draw[gray!30!black,fill=gray!60,rounded corners=2,opacity=0.5]
     (2.05,3.1) rectangle (2.75,4);

\draw[gray!30!black,fill=gray!60,rounded corners=2,opacity=0.5]
     (5.25,3.1) rectangle (5.95,4);

\draw[gray!30!black,fill=gray!60,rounded corners=2,opacity=0.5]
     (6.05,2.3) rectangle (6.75,3.2);

\vertex{0}{2}{v1};
\vertex{2}{0}{v2};
\vertex{1.6}{3.9}{v3};
\vertex{2.4}{3.4}{v4};
\vertex{4}{0}{v5};
\vertex{4}{4}{v6};
\vertex{6}{0.5}{v7};
\vertex{5.6}{3.4}{v8};
\vertex{6.4}{2.6}{v9};

\vertexLabel[left]{v1}{$v_1$}
\vertexLabel[below]{v2}{$v_2$}
\vertexLabel[above]{v3}{$v_3$}
\vertexLabel[above]{v4}{$v_4$}
\vertexLabel[below]{v5}{$v_5$}
\vertexLabel[above]{v6}{$v_6$}
\vertexLabel[below]{v7}{$v_7$}
\vertexLabel[above]{v8}{$v_8$}
\vertexLabel[above]{v9}{$v_9$}

\path (v1) edge [e1] (v2);
\path (v1) edge [e1] (v3);
\path (v1) edge [e1] (v4);
\path (v2) edge [e1] (v5);
\path (v2) edge [e1] (v6);
\path (v3) edge [e1] (v4);
\path (v3) edge [e1] (v5);
\path (v3) edge [e1] (v6);
\path (v4) edge [e1] (v5);
\path (v4) edge [e1] (v6);
\path (v5) edge [e1] (v6);
\path (v5) edge [e1] (v7);
\path (v5) edge [e1] (v8);
\path (v5) edge [e1] (v9);
\path (v6) edge [e1] (v7);
\path (v6) edge [e1] (v8);
\path (v6) edge [e1] (v9);
\path (v7) edge [e1] (v8);
\path (v7) edge [e1] (v9);

\begin{scope}[xshift=8.5cm,yshift=0.3cm]

\vertex{0}{2}{v1};
\vertex{1}{1}{v2};
\vertex{2}{0}{v3};
\vertex{4}{0}{v4};
\vertex{5}{1}{v5};
\vertex{7}{1}{v6};
\vertex{8}{1}{v7};
\vertex{9}{0}{v8};
\vertex{11}{0}{v9};

\vertexLabel[below]{v1}{$v_1$}
\vertexLabel[below]{v2}{$v_2$}
\vertexLabel[below]{v3}{$v_3$}
\vertexLabel[below]{v4}{$v_4$}
\vertexLabel[below]{v5}{$v_5$}
\vertexLabel[below]{v6}{$v_6$}
\vertexLabel[below]{v7}{$v_7$}
\vertexLabel[below]{v8}{$v_8$}
\vertexLabel[below]{v9}{$v_9$}

\node[draw, gray!30!black,fill=gray!50,rounded corners=3,opacity=0.5, fit={(2,0.7) (4,1.3)}, inner sep=0.25pt, label=center:\small\phantom{j}Series\phantom{j}] (A) {};
\node[draw, gray!30!black,fill=gray!50,rounded corners=3,opacity=0.5, fit={(1,1.7) (3,2.3)}, inner sep=0.25pt, label=center:\small\phantom{j}Parallel\phantom{j}] (B) {};
\node[draw, gray!30!black,fill=gray!50,rounded corners=3,opacity=0.5, fit={(5,1.7) (7,2.3)}, inner sep=0.25pt, label=center:\small\phantom{j}Series\phantom{j}] (C) {};
\node[draw, gray!30!black,fill=gray!50,rounded corners=3,opacity=0.5, fit={(9,0.7) (11,1.3)}, inner sep=0.25pt, label=center:\small\phantom{j}Parallel\phantom{j}] (D) {};
\node[draw, gray!30!black,fill=gray!50,rounded corners=3,opacity=0.5, fit={(8,1.7) (10,2.3)}, inner sep=0.25pt, label=center:\small\phantom{j}Series\phantom{j}] (E) {};
\node[draw, gray!30!black,fill=gray!50,rounded corners=3,opacity=0.5, fit={(3.6,3.75) (7.4,4.35)}, inner sep=0.25pt, label=center:{\small\phantom{j}Prime\phantom{aaaaaaaaa}}] (F) {};
\draw (A)--(v4);
\draw (A)--(v3);
\draw (B)--(A);
\draw (B)--(v2);
\draw (D)--(v8);
\draw (D)--(v9);
\draw (C)--(v5);
\draw (C)--(v6);
\draw (E)--(v7);
\draw (E)--(D);
\draw (F)--(v1);
\draw (F)--(B);
\draw (F)--(C);
\draw (F)--(E);

\minivertex{5.5}{4.05}{w1};
\minivertex{6}{4.05}{w2};
\minivertex{6.5}{4.05}{w3};
\minivertex{7}{4.05}{w4};

\path (w1) edge [e1] (w4);

\end{scope}

\end{tikzpicture}
    }%
    \caption{Example of a graph and its modular decomposition tree.}
    \label{fig:mod-decomp}
\end{figure}

The \emph{modular-width} $\mw(G)$ of $G$~\citep{modular-width} is the maximum number of children of a prime node in the modular decomposition tree of $G$, or $2$, if no prime node exists. 

\begin{definition}[$\mathcal{G}$-modular cardinality {\citep[Definition 1]{modular-partitions-MFCS23}}]

Let $\mathcal{G}$ be a graph class that contains the graph with one vertex as a member. For a graph $G$ (not necessarily in $\mathcal{G}$), a module $M$ of
$G$ is a \emph{$\mathcal{G}$-module} if $G[M]$ belongs to $\mathcal{G}$. A modular partition $\mathcal{M} \coloneqq \{M_1, \dots , M_k\}$ of a graph
$G$ is called a \emph{$\mathcal{G}$-modular partition} if each $M_i$ is a $\mathcal{G}$-module. The \emph{$\mathcal{G}$-modular cardinality of
$G$}, denoted by $\Gmc(G)$, is the minimum cardinality of a $\mathcal{G}$-modular partition of $G$.
\end{definition}

\paragraph{Thinness.} A graph $G$ is \emph{$k$-thin} if there exists a strict total order $<$ on $V(G)$ and a partition $S$ of $V(G)$ into
$k$ classes such that, for each $u, v, w \in V(G)$ with $u < v <
w$, if $u$ and $v$ belong to the same class and $(u,w) \in
E(G)$, then $(v,w) \in E(G)$. The three vertices $u, v, w$ satisfying those conditions are said to form a \emph{consistent triple}. An order and a partition
satisfying those properties are said to be \textit{consistent} with $G$. We
call the tuple $(<, S)$ a \emph{consistent solution} for $G$. The minimum $k$ such that $G$ is
$k$-thin is called the \textit{thinness} of $G$, and denoted by
$\thin(G)$. A consistent solution for $G$ that uses $\thin(G)$ classes is said to be \emph{optimal}. The \textsc{Thinness} problem takes as input a graph and a natural number $k$ and consists of deciding if the input graph has thinness at most $k$.

\paragraph{Simultaneous interval number.} A \emph{simultaneous interval representation} of a graph $G$ is the result of assigning an interval and a subset of \emph{labels} from $\{1,\dots,d\}$ for $d \in \N$ to each vertex in $V(G)$ such that two vertices are adjacent in $G$ if and only if both their intervals and their sets of labels intersect. We say that $d$ is the \emph{size} of a simultaneous interval representation. The \emph{simultaneous interval number} $\si(G)$ of a graph $G$ is the smallest size of a simultaneous interval representation of $G$. The \textsc{Simultaneous Interval Number} problem consists in determining, for a graph $G$ and an integer $d \geq 0$, if $G$ admits a $d$-simultaneous representation or not.

\paragraph{Neighborhood diversity.} We say that two vertices $v$ and $v'$ of a graph $G$ have the same \emph{neighborhood type} if and only if $N(v) \setminus \{v'\} = N(v') \setminus \{v\}$ \citep[Definition 1]{neighborhood-diversity}. In other words, $v$ and $v'$ have the same neighborhood type if and only if they induce a module in $G$. A graph $G$ has \emph{neighborhood diversity} at most $w$ if there exists a partition of $V(G)$ into at most $w$ sets such that all the vertices in each set have the same neighborhood type. We will call this partition a \emph{neighborhood partition} of the vertices of $G$. The neighborhood diversity of $G$, denoted $\nd(G)$, is the smallest such $w$ \citep[Definition 2]{neighborhood-diversity}.

\paragraph{Twin-cover.} $X \subseteq V(G)$ is a \emph{twin-cover} of $G$ if for every edge $(a,b) \in E(G)$ either $a \in X$, or $b \in X$, or $a$ and $b$ are \emph{twins}, i.e. $N[a]=N[b]$. The \emph{twin-cover number} of $G$, $\tc(G)$, is the minimum size of a twin-cover of $G$ \citep[Definition 3.1]{twin-cover}.

\paragraph{Vertex cover.} Let $G$ be a graph. A subset $X$ of $V(G)$ is a vertex cover of $G$ if for every edge $(u, v)$ of $E(G)$ at least one of $\{u, v\}$ belongs to $X$. The \emph{vertex cover number} of $G$, $\vc(G)$, is the minimum size of a vertex cover of $G$. 

\paragraph{Parameterized complexity.} The field of \emph{parameterized complexity} seeks to find efficient algorithms for problems when the input is accompanied by a number that bounds a \emph{parameter} of the input. Algorithms that take polynomial time when that parameter is fixed are of special interest. Formally, a \emph{parameterized problem} is a language $Q \subseteq \Sigma^* \times N$, and such a problem is \emph{fixed-parameter tractable} (\FPT) if there is an algorithm that decides membership of an instance $(x, k)$ in time $f(k)\cdot\abs{x}^{\O(1)}$ for some computable function~$f$. A problem is \emph{slice-wise polynomial} (\XP) if there is an algorithm that decides membership of an instance $(x, k)$ in time $f(k)\cdot \abs{x}^{f(k)}$ for some computable function $f$.

\paragraph{Kernelization.} A \emph{kernelization} (or \emph{kernel}) for a parameterized problem $Q$ is a po\-ly\-no\-mial-time algorithm that transforms each instance $(x, k)$ into an instance $(x', k')$ such that:
\begin{itemize}
\item $(x', k') \in Q$ if and only if $(x,k) \in Q$, and
\item $\abs{x'}, k' \leq g(k)$ for some computable function $g$.
\end{itemize}
We say that $g(k)$ is the \emph{size} of the kernel.

\section{Reducing modules without changing the thinness}
\label{sec:reducing-modules}

We begin by showing that if a sufficiently large group of vertices have the same neighborhood, some of them can be removed from a graph without affecting its thinness. This will be key in proving the kernelizations results for \textsc{Thinness} in \cref{sec:parameterizations}, and when showing the linear time algorithm to compute the $\interval$-modular cardinality in \cref{prop:calculate-interval-module-partition}.

We state here a result by \citet{thinness-of-product-graphs} on the thinness of a graph after contracting a complete module.

\begin{lemma}[{\citep[Theorem 35]{thinness-of-product-graphs}}]
    \label[lemma]{prop:contract-complete-module}
    Let $H$ be a module of a graph $G$, and $\contract{G}{H}$ be the graph obtained by contracting $H$ into a vertex. If $H$ is complete, then $\thin(G) = \thin(\contract{G}{H})$.
\end{lemma}

A similar lemma holds for modules that induce interval graphs that are not complete.
\begin{lemma}
    \label[lemma]{prop:contract-interval-module}
    Let $H$ be a module of a graph $G$ such that $G[H]$ is an interval graph and $H$ is not complete. Let $G'$ be the result of removing from $G$ all vertices belonging to $H$ except for two nonadjacent vertices. Then $\thin(G) = \thin(G')$.
\end{lemma}
\begin{proof}
    As $G'$ is an induced subgraph of $G$, we have $\thin(G) \geq \thin(G')$. It remains to prove $\thin(G) \leq \thin(G')$.
    
    Starting from a consistent solution $(<', S')$ on $\thin(G')$ classes of the vertices of $G'$, we will construct a consistent solution $(<, S)$ of the vertices of $G$ on no more than $\thin(G')$ classes, and thus prove $\thin(G) \leq \thin(G')$.
    
    Let $h_1$ and $h_2$ denote the two remaining vertices belonging to $H$ in $G'$ such that $h_1 <' h_2$. Take $<_H$ to be an ordering of the vertices of $H$ consistent with the partition $\{H\}$, which exists because $H$ is an interval graph. We construct solution $(<, S)$ by removing vertex $h_2$ from $(<', S')$ and replacing the vertex $h_1$ in $<'$ with all the vertices in $H$ in the order defined by $<_H$. Additionally, we assign all the vertices of $H$ to the same class as $h_1$ in $S'$. 
    
    We claim that $(<, S)$ is a consistent solution for $G$. We will prove this by analyzing every possible triple $u < v < w$ of vertices in $G$, and showing that they form a consistent triple. We distinguish several cases
    depending on which of the three vertices belong to $H$:

    \begin{itemize}
        \item \underline{None of $u$, $v$, $w$ belongs to $H$:} This triple is consistent in $(<', S')$, and thus is also consistent in $(<, S)$.
        \item \underline{Only one of $u$, $v$, $w$ belongs to $H$:} If this triple were inconsistent, then replacing the vertex that belongs to $H$ in $\{u, v, w\}$ with $h_1$ would form an inconsistent triple in $(<', S')$, which is a contradiction to the fact that $(<', S')$ is a consistent solution.
        \item \underline{$\{u, v\} \subseteq H$ and $w \not \in H$:} As $H$ is a module, vertex $w$ is either adjacent to both $u$ and $v$, or nonadjacent to both of them. Thus, this is a consistent triple.
        \item \underline{$u \in H$, $v \not \in H$, and $w \in H$:} This cannot happen, as the vertices in $H$ are consecutive in $<$.
        \item \underline{$u \not \in H$ and $\{v, w\} \subseteq H$:} We know that in $G'$ we have $u <' h_1 <' h_2$. As $(<', S')$ is a consistent solution, vertices $\{u, h_1, h_2\}$ must form a consistent triple. As $h_2$ is not adjacent to $h_1$, either
            \begin{itemize}
                \item $u$ is not adjacent to $h_2$, which means that $u$ is also not adjacent to any vertex in the module $H$ (vertex $w$ in particular); or
                \item $u$ and $h_1$ do not share a class, which means that $u$ and $v$ do not share a class either.
            \end{itemize}
        In both cases, $\{u, v, w\}$ forms a consistent triple.
     \end{itemize}

    This proves that $(<, S)$ is a consistent solution with at most $\thin(G')$ classes for the vertices in $V(G)$, and thus $\thin(G) \leq \thin(G')$. Combining this result with $\thin(G) \geq \thin(G')$ we have that $\thin(G) = \thin(G')$.
\end{proof}

In \cref{prop:contract-interval-module}, the choice of replacing $H$ with two independent vertices is superfluous: in fact, replacing the module with any non-complete interval graph would maintain the thinness of the whole graph. Recalling that interval graphs are exactly the graphs of thinness one, it seems sensible to ask whether \cref{prop:contract-interval-module} could be generalized to modules of greater thinness. In particular, \emph{is the thinness of a graph $G$ maintained when modules of thinness two are replaced by any other graph $H'$ such that $\thin(H') = 2$?} 
The answer is negative, and we defer the construction of a counterexample to \cref{sec:counterexample}.

Combining \cref{prop:contract-complete-module} and \cref{prop:contract-interval-module}, we see that if a module $H$ induces an interval graph, almost all vertices of $H$ can be removed without altering the thinness of the graph. Given a \textsc{Thinness} instance and a partition of the input graph into $\interval$-modules, this can be used in a preprocessing step to reduce the size of the input graph, and thus obtain a kernel. We thus need a way to efficiently calculate the minimum number of modules in an $\interval$-modular partition of the graph, and this is done in the next section.

\section{Linear-time algorithms for \texorpdfstring{$\cluster$}{clusters}-modular cardinality and\\ \texorpdfstring{$\interval$}{intervals}-modular cardinality}
\label{sec:modular-cardinality}

As mentioned in the introduction, the article that introduces the $\G$-modular cardinality~\citep{modular-partitions-MFCS23}\footnote{Full proofs can be found in~\citep{modular-partitions-arxiv}.} presents a polynomial-time algorithm to obtain it for an arbitrary graph when $\G$ meets some conditions, which the classes $\cluster$ and $\interval$ meet. Their algorithm does not explicitly run in linear time when $\G$ corresponds to the class of interval or cluster graphs. Thus, we present similar linear-time algorithms for these two parameters that utilize the modular decomposition of the input graph.

\begin{theorem}
\label{prop:calculate-cluster-module-partition}
    An optimal $\cluster$-modular partition of a graph $G$ can be computed in linear time.
\end{theorem}
\begin{proof}
\newcommand{\Shead}{C_1}
\newcommand{\Stail}{\mathsf{tail}}
First, we compute the modular decomposition tree $T$ of $G$, which can be done in linear time~\citep{modular-decomposition-linear-cournier,modular-decomposition-linear-mcconnell,modular-decomposition-linear}. Then, for each node $G_t \in T$ (which is an induced subgraph of $G$), we compute an optimal $\cluster$-modular partition $S_t$ of the vertices of $G_t$ depending on the type of node $G_t$. Let $\mathcal{C} \coloneqq \{G_1, \dots, G_h\}$ be the children of $G_t$, if there are any. 
    \begin{itemize}
        \item \textbf{Leaf}: The only class in $S_t$ is formed by the only vertex in $G_t$. This vertex is trivially a cluster module, and the partition is unique, thus optimal.
        \item \textbf{Parallel or Series}: Let $\mathcal{K} \subseteq \mathcal{C}$ be the children of $G_t$ that are complete graphs. We define $S_t \coloneqq \{\Shead\} \cup \Stail$, where
        \begin{align*}
            \Shead &\coloneqq \bigcup_{G_i \in \mathcal{K}} V(G_i)\\
            \Stail &\coloneqq \bigcup_{G_i \in \mathcal{C} \setminus \mathcal{K}} S_{t_i}. 
        \end{align*}
        In other words, we take the vertices of the child nodes $t_i$ for which $G_{t_i}$ is a complete graph, and join them into one class to build the first class of $S_t$. Then, we add the classes of all the other partitions as they are.

        As shown by \citet{modular-decomposition}, the only modules present in $G_t$ are the possible unions of elements of $\mathcal{C}$, and the submodules of elements of $\mathcal{C}$. That is to say, $G_t$ has no modules that contain some but not all vertices of some $G_i$, and some but not all vertices of some $G_j$, with $i \neq j$.

        Let us see first that this is a valid $\cluster$-modular partition. By the previous remark, the class $\Shead$ and every element of $\Stail$ are modules. If $G_t$ is a parallel node, $\Shead$ is a union of complete graphs, meaning, a cluster. If $G_t$ is a series node, $\Shead$ is itself a complete graph, and so it is also a cluster. The elements of $\Stail$ are clusters, as they are part of the $\cluster$-modular partitions of the children of~$G_t$. Moreover, all vertices of $G_t$ are present in some set of the partition, and thus $S_t$ is a valid $\cluster$-modular partition.

        Next, let us see that it is optimal. Suppose there is an optimal partition $Q_t$ of the vertices of $G_t$ that uses fewer classes than $S_t$. As noted before, the classes of $Q_t$ must each be either a union of elements of $\mathcal{C}$, or submodules of elements of $\mathcal{C}$.
        
        Suppose that some class $Q^i_t$ of $Q_t$ is a union of more than one element of $\mathcal{C}$, and that one of these elements $G_j$ is not complete. We distinguish two cases:
        \begin{itemize}
            \item \emph{$G_t$ is a parallel node}: Every element of $\mathcal{C}$ is a connected graph. As $G_j$ is not complete, it cannot be part of a cluster.
            \item \emph{$G_t$ is a series node}: Let $G_l$ be another element of $\mathcal{C}$ present in $Q^i_t$. Every vertex of $G_j$ is adjacent to every vertex of $G_l$. As $G_j$ is not complete, there are two nonadjacent vertices $v$ and $w$ in $V(G_j)$ that are each adjacent to a vertex $u$ in $G_l$. Thus, $v$, $w$, and $u$ cannot be part of the same cluster. 
        \end{itemize}
        These two cases show that the previous scenario cannot happen, and so every class of $Q_t$ is either a union of elements of $\mathcal{C}$ that are complete graphs, or submodules of elements of $\mathcal{C}$. Moreover, there cannot be more than one class of $Q_t$ that is a union of elements of $\mathcal{C}$, because otherwise we could join those classes into only one, reducing the size of $Q_t$. Also, there cannot be a proper submodule of an element $G_i$ of $\mathcal{C}$ that is a complete graph as a class in $Q_t$, as otherwise we could join that submodule with the rest of $G_i$ to reduce the size of $Q_t$. Thus, the union of the elements of $\mathcal{C}$ that are complete graphs defines a class in $Q_t$. Finally, as each other class is a submodule of an element of $\mathcal{C}$, the classes that have some vertex of a given graph $G_i$ must form a partition of the vertices of $G_i$, meaning, they should be a $\cluster$-modular partition of $G_i$. Partition $Q_t$ then contains an optimal cluster partition of each $G_i$, and necessarily has the same amount of classes as $S_t$.
        
        \item \textbf{Prime}: We define $S_t$ to be the union of $S_i$ for every $1 \leq i \leq h$.
        
        As shown by \citet{modular-decomposition}, the proper modules of $G_t$ are exactly the submodules of the associated graphs of children of $G_t$. Thus, an optimal $\cluster$-modular partition of $G_t$ must contain an optimal $\cluster$-modular partition of every child $G_i$, and so $S_t$ is optimal.
    \end{itemize}

    To achieve linear time on this algorithm, the partitions $S_t$ are not actually computed. Instead, we set a flag on each child $t_i$ of a series or parallel node $t$ such that $G_{t_i}$ belongs to $\Shead$. Recognizing these children can be done with a simple bottom-up traversal of the tree that marks each node whose associated graph is complete (i.e., a parallel node whose children are leaves). Finally, a top-down traversal of the tree allows us to propagate to which class each of the leaf nodes belongs. These operations can be performed in time linear in the number of nodes in the tree, and thus this algorithm runs in linear time.
\end{proof}

We present two lemmas that will be useful to compute an optimal $\interval$-modular partition.

\begin{lemma}[Thinness of the union {\citep[Theorem 14]{tesis-diego}}]
    \label[lemma]{prop:thinness-union}
    Let $G_1$ and $G_2$ be graphs. Then $\thin(G_1 \cup G_2) = \max\{\thin(G_1), \thin(G_2)\}$.
\end{lemma}

\begin{lemma}[Thinness of the join {\citep[Theorem 27]{thinness-of-product-graphs}}]
    \label[lemma]{prop:thinness-join}
    Let $G_1$ and $G_2$ be graphs. If $G_1$ is complete, then $\thin(G_1 \lor G_2) = \thin(G_2)$. If neither $G_1$ nor $G_2$ are complete, then $\thin(G_1 \lor G_2) = \thin(G_1) + \thin(G_2)$.
\end{lemma}

In particular, \cref{prop:thinness-join} implies that the join of two graphs is an interval graph if and only if one of them is an interval graph, and the other one is a complete graph.
\begin{theorem}
    \label{prop:calculate-interval-module-partition}
    An optimal $\interval$-modular partition of a graph $G$ can be computed in linear time.
\end{theorem}
\begin{proof}
\newcommand{\Stail}{\mathsf{tail}}
    The proof follows the structure of the one for \cref{prop:calculate-cluster-module-partition}. We compute the modular decomposition tree $T$ of $G$. Then, for each node $G_t \in T$ (which is an induced subgraph of $G$), we compute an optimal $\interval$-modular partition $S_t$ of the vertices of $G_t$ depending on the type of node $G_t$. Let $\mathcal{C} \coloneqq \{G_1, \dots, G_h\}$ be the children of $G_t$, if there are any, and let $S_1,\dots, S_h$ be their respective optimal $\interval$-modular partitions.
    \begin{itemize}
        \item \textbf{Leaf}: The only class in $S_t$ is formed by the only vertex in $G_t$. This vertex is trivially an interval module, and the partition is unique, thus optimal.
        \item \textbf{Parallel}: We build $S_t$ by taking the union of the children that induce interval graphs as class $C_1$, and taking the union of the partitions of all the other children. More formally, if $\mathcal{I}$ is the subset of the children $\mathcal{C}$ that induce interval graphs: 
        \begin{align*}
            C_1 &\coloneqq \bigcup_{G_i \in \mathcal{I}}V(G_i)\\
            S_t &\coloneqq \{C_1\} \cup \bigcup_{G_i \in \mathcal{C}\setminus\mathcal{I}}S_i.
        \end{align*}
        
        Note that $S_t$ is in fact an $\interval$-modular partition:
        \begin{itemize}
            \item All the classes induce interval graphs, as $C_1$ is a union of interval graphs (and thus by \cref{prop:thinness-union} it is an interval graph), and all the other classes belong to $\interval$-modular partitions of the children of $G_t$, and thus are also interval graphs.
            \item Each class is a module, a $C_1$ is a union of modules, and all the other classes belong to $\interval$-modular partitions of induced subgraphs.
        \end{itemize}

        Let us see that $\abs{S_t}$ is minimum. Again, as shown by \citet{modular-decomposition}, the only modules of $G_t$ are unions of children of $G_t$, or submodules of children of $G_t$. Thus, in every $\interval$-modular partition $S'_t$ of $G_t$,
        \begin{enumerate}
            \item\label{item:union} by \cref{prop:thinness-union}, each class that induces a union of disjoint graphs must have interval graphs as connected components; and
            \item\label{item:submodules} every other class must be a submodule of one of the children of $G_t$.
        \end{enumerate}
        Partition $S_t$ minimizes both the number of classes of type \ref{item:union} and the number of classes of type \ref{item:submodules}, and is therefore of minimum size.
        \item \textbf{Series}: We define $S_t$ as $\{C_1\} \cup \Stail$, where
        \begin{itemize}
            \item $C_1$ contains all vertices belonging to complete children of $G_t$, and all vertices of one arbitrary child that is an interval graph if it exists; and
            \item $\Stail$ is the union of optimal partitions for the rest of the children of $G_t$.
        \end{itemize}

        If $C_1$ is empty, we just define $S_t \coloneqq \Stail$.

        By \cref{prop:thinness-join}, class $C_1$ induces an interval graph. Every class in $\Stail$ also induces an interval graph, as they belong to $\interval$-modular partitions of the children of $G_t$. On the other hand, notice that all classes induce modules of $G_t$. Thus, $S_t$ is an $\interval$-modular partition.

        Take another $\interval$-modular partition $S'_t$ of $G_t$. We will see that $\abs{S'_t} \geq \abs{S_t}$.
        
        As said earlier, the modules of $G_t$ are unions of children of $G_t$, or submodules of children of $G_t$. For a module that is a join of two graphs to be an interval graph, it must be the join of a complete graph and an interval graph by \cref{prop:thinness-join}. Thus, if two or more children of $G_t$ are joined in a class of $S'_t$, at most one of them is not a complete graph, and it must also be an interval graph. If there is more than one class in $S'_t$ that induces a join of graphs, we can change the classes of all vertices belonging to children that induce complete graphs to just one of those classes, and have an $\interval$-modular partition with no more classes than $S'_t$.

        Furthermore, $S'_t$ is forced to contain $\interval$-modular partitions of the children of $G_t$ that do not induce interval graphs. Therefore, $S_t$ contains no more classes than $S'_t$.

        \item \textbf{Prime}: If $G_t$ is an interval graph, define $S_t \coloneqq \{V(G_t)\}$. Otherwise, $S_t \coloneqq \bigcup_{G_i \in \mathcal{C}}S_i$.

        It is clear that, in both cases, $S_t$ is an $\interval$-modular partition. Also, as shown by \citet{modular-decomposition}, the only proper submodules of a prime node are the modules of children of the node. Thus, if $G_t$ is not an interval graph, the partition must be the union of $\interval$-modular partitions of the children of $G_t$, and therefore $S_t$ is optimal.

        It remains to show how to determine if $G_t$ is an interval graph or not. In principle, one could run a linear-time algorithm that decides if a graph is interval or not on $G_t$, for example the one by \citet{B-L-interval}. The issue is that running a linear time algorithm for each of the prime nodes of the modular decomposition yields a runtime complexity of $\O(n^2)$, which is more than what we want.

        To remedy this, we will utilize the lemmas presented in \cref{sec:reducing-modules} to reduce the number of vertices in $G_t$ without changing the thinness. First, if any one of the children of $G_t$ is not an interval graph, then, as the thinness is a hereditary property, graph $G_t$ is also not an interval graph. Otherwise, we create a graph $G'_t$ by contracting every child that is a complete graph into a single vertex, and contracting every other child that is an interval graph into two nonadjacent vertices. By \cref{prop:contract-complete-module,prop:contract-interval-module}, graph $G'_t$ has the same thinness as $G_t$. In particular, $G'_t$ is an interval graph if and only if $G_t$ is an interval graph. Thus, we finish by running a linear-time algorithm to decide if $G'_t$ is an interval graph or not.

        As the number of nodes in a modular decomposition tree is $\O(n)$, and the number of vertices in $G'_t$ is at most twice the number of children of $G_t$, this algorithm takes linear time on the size of $G$, taking into account that we need to store a flag for each node that signals if a node is a complete graph or not.\qedhere
     \end{itemize}
\end{proof}

\section{Parameterizations for computing the thinness}\label{sec:parameterizations}

We are now ready to present the kernels for {\sc Thinness} parameterized by the $\interval$-modular cardinality, neighborhood diversity, twin-cover, and vertex cover.

\begin{theorem}
    \label{prop:interval-modular-cardinality-parametrization}
    For every graph $G$, $\thin(G) \leq 2\cdot\intervalmc(G)$. In addition, {\sc Thinness} admits a linear kernel when parameterized by the $\interval$-modular cardinality of the input graph.
\end{theorem}
\begin{proof}
    We reduce an instance $(G, t, k)$ of \textsc{Thinness} parameterized by $\interval$-modular cardinality (where $t$ is the thinness parameter and $k$ is the $\interval$-modular cardinality of $G$) to an instance $(G', t', k')$ in time $n^{\mathcal{O}(1)}$ such that $t' = t$, $k' = k$, and $\abs{V(G')} \leq 2k$.

    We first compute an optimal $\interval$-modular partition $P$ of $G$ in linear time by \cref{prop:calculate-interval-module-partition}. Now, we contract each part of $P$ that induces a complete graph into a single vertex each, and contract every other part of $P$ into two nonadjacent vertices each, obtaining graph $G'$. By \cref{prop:contract-complete-module} and \cref{prop:contract-interval-module} we have that $\thin(G') = \thin(G)$.
    
    We now have a graph $G'$ such that $\abs{V(G')} \leq 2k$ and $\thin(G') = \thin(G)$ which serves as our kernel, defining $k'\coloneqq k$. Note that all operations we performed on the graph can be done in linear time, and so the reduction is polynomial.
    This also shows that $\thin(G) \leq 2\cdot\intervalmc(G)$, as desired.
\end{proof}

Note that every cluster graph is also an interval graph, and thus $\intervalmc(G) \leq \clustermc(G)$ for every graph $G$. Therefore, although the upper bounds we present on \cref{prop:neighborhood-partition-is-cluster-modular-partition,lemma:twin-cover-hcupn} are on the $\cluster$-modular cardinality instead of the $\interval$-modular cardinality, the results transfer directly to the latter parameter.

The following lemmas show that the neighborhood partition of a graph is also a $\cluster$-modular partition.

\begin{lemma}[{\citep[Theorem 7.1]{neighborhood-diversity}}]
    \label[lemma]{lemma:neighborhood-partition-clique-or-independent}
    Let $G$ be a graph with neighborhood diversity $k$, and let $V_1, \dots, V_k$ be a neighborhood partition of width $k$ of the vertices of $G$. Then each $V_i$ induces either a clique or an independent set.
\end{lemma}

\citet{modular-partitions-MFCS23} state in their article that the $\cluster$-modular cardinality generalizes the neighborhood diversity. We present a proof of this fact here for completeness.

\begin{lemma}
    \label[lemma]{prop:neighborhood-partition-is-cluster-modular-partition}
    For every graph $G$, $\clustermc(G) \leq \nd(G)$. Moreover, a neighborhood partition of the vertices of $G$ is also a $\cluster$-modular partition of $G$.
\end{lemma}
\begin{proof}
    Let $V_1,\dots,V_k$ be a neighborhood partition of width $k$ of the vertices of $G$. By \cref{lemma:neighborhood-partition-clique-or-independent}, each $V_i$ induces either a clique or an independent set. In either case, $G[V_i]$ is a cluster graph. Also, each $V_i$ is a module, since all the vertices in $V_i$ have the same neighborhood type, and so must all have the same neighbors outside $V_i$. Therefore, $V_1,\dots, V_k$ is a valid $\cluster$-modular partition of width $k$ of $G$, and so $\clustermc(G) \leq k$.
\end{proof}

On the other hand, the neighborhood diversity of a graph can be arbitrarily larger than the $\cluster$-modular cardinality, as the graph $cK_2$ has $\cluster$-modular cardinality 1 but neighborhood diversity $c$.

Combining \cref{prop:interval-modular-cardinality-parametrization} with \cref{prop:neighborhood-partition-is-cluster-modular-partition} we obtain the following result.

\begin{theorem}
\label{prop:neighborhood-diversity-parametrization}
    {\sc Thinness} is \FPT when parameterized by the neighborhood diversity of the input graph. Moreover, the problem admits a linear kernel.
\end{theorem}


The following lemma shows that the $\cluster$-modular cardinality of a graph generalizes the twin-cover number.

\begin{lemma}
    \label[lemma]{lemma:twin-cover-hcupn}
    For every graph $G$ with twin-cover number $k$, $\cm(G) \leq 2^k + k$.
\end{lemma}
\begin{proof}
   Let $X$ be a twin-cover of $G$ of size $k$. We construct a partition $P$ with parts $P_1, \dots, P_{2^k + k}$ of the vertices of $G$ and prove that it is a $\cluster$-modular partition.
    
    We first put each vertex $x_i \in X$ in the set $P_{2^k+i}$. Then, we partition the vertices not in $X$ by the set of vertices in $X$ that they are adjacent to. As there are $k$ vertices in $X$, there are $2^k$ possible subsets of $X$. We assign these sets to $P_1, \dots, P_{2^k}$. 
    
    Sets $P_{2^k+1},\dots,P_{2^k+k}$ have only one vertex, and so are cluster modules.
    It remains to show that sets $P_1,\dots,P_{2^k}$ are cluster modules.

    We first prove that they are modules. By definition, every vertex in $X$ is either adjacent to every vertex in $P_i$ or to none. 
    
    We will now show that if $v \in P_i$ and $w \in P_j$ with $1 \leq i < j \leq 2^k$, there $v$ and $w$ are not adjacent. By definition of the first $2^k$ sets, the neighbors of $v$ and $w$ in $X$ are not the same, as they belong to different sets. This means, in particular, that $v$ and $w$ are not twins, and then, as none of them belongs to the twin-cover $X$, they cannot be adjacent.

    Now, we prove that sets $P_1,\dots,P_{2^k}$ are clusters. To see that, it is sufficient to show that if $a,b \in P_i$ are adjacent, then $N[a] = N[b]$. As neither $a$ nor $b$ are in $X$, they are twins, and so $N[a] = N[b]$.
\end{proof}

The twin-cover number, similarly to the neighborhood diversity, can be arbitrarily larger than the $\cluster$-modular cardinality. The class of complete bipartite graphs $K_{n,n}$ serves as an example here, as those graphs have $\cluster$-modular cardinality two and twin-cover number $n$.

\begin{theorem}
    \label{prop:twin-cover-parametrization}
    {\sc Thinness} is \FPT when parameterized by the size of a minimum twin-cover of the input graph.
\end{theorem}
\begin{proof}
    Let $G$ be a graph with twin-cover number $k$. By \cref{lemma:twin-cover-hcupn}, $\cm(G) \leq 2^k + k$. We thus use the algorithm from \cref{prop:interval-modular-cardinality-parametrization} to obtain a kernel of size $\O(2^k)$.
\end{proof}

Note that the kernel presented above is not polynomial, as the bound on the $\cluster$-modular cardinality given a twin-cover of size $k$ is $\O(2^k)$.

As every vertex cover is also a twin-cover, we have the following corollary.

\begin{corollary}
\label{prop:vertex-cover-parametrization}
    {\sc Thinness} is \FPT when parameterized by the size of a minimum vertex cover of the input graph.
\end{corollary}

\section{Parameterizations for computing the simultaneous interval number}
\label{sec:sim-int}

We prove here the bounds shown in \cref{fig:diagram}, as well as some algorithmic results involving the computation of the simultaneous interval number.
We first formally define the problem. 

\defproblem{Simultaneous Interval Number}{A graph $G$ and an integer $d \geq 0$.}{Is $\si(G) \leq d$? I.e., Does $G$ admit a $d$-simultaneous interval representation?}

It was proved by \citet{Milanic-sim} that $\si(G)=0$ if and only if the graph $G$ is edgeless and that $\si(G)\leq 1$  if and only if $G$ is an interval graph. Indeed, interval graphs can be represented by their own interval representation and labeling all vertices with the set $\{1\}$.

We prove here results about modules and simultaneous interval representations.
\begin{lemma}
\label[lemma]{prop:si-contract-complete-module}
    Let $H$ be a module of a graph $G$ having neighbors outside $H$, and let $\contract{G}{H}$ be the graph obtained by contracting $H$ into a vertex. If $H$ is complete, then $\si(G) = \si(\contract{G}{H})$. Moreover, $G$ admits an optimal simultaneous interval representation such that all the intervals and label sets representing vertices of $H$ are identical. 
\end{lemma}
\begin{proof}
    Suppose $H$ was contracted to a vertex $v$, and let us consider a $d$-simultaneous interval representation of $\contract{G}{H}$. We can extend it to a $d$-simultaneous interval representation of $G$ by representing every other vertex of $H$ with the same interval and the same label set as $v$, which is not empty because $v$ has neighbors in $\contract{G}{H}$.
\end{proof}

We can generalize this to the following property.
\begin{lemma}
\label[lemma]{prop:si-module-pairwise-intersecting}
    Let $H$ be a module of a graph $G$. If $G$ admits a $d$-simultaneous interval representation such that the intervals representing vertices of $H$ are pairwise intersecting, then $G$ admits a $d$-simultaneous interval representation such that the intervals representing vertices of $H$ are identical. In this case, the label sets have to be disjoint whenever two vertices are non-adjacent and intersecting otherwise. 
\end{lemma}
\begin{proof}
    Suppose we have a $d$-simultaneous interval representation of $G$ such that the intervals representing vertices of $H$ are pairwise intersecting. Let $x$ and $y$ be, respectively, the rightmost left endpoint and the leftmost right endpoint of the intervals representing vertices of $H$. Since the intervals are pairwise intersecting, $x < y$. Since $H$ is a module, every vertex that is adjacent to a vertex of $H$ is adjacent to all of them, thus its interval intersects the interval $(x,y)$. So, no intersection is lost (and no intersection is gained) by keeping the label sets and shortening each interval representing a vertex of $H$ to $(x,y)$.
\end{proof}

Now, we describe the situation for modules that induce interval graphs that are not complete, which in particular include the independent sets.
For complete bipartite graphs, it was proved by \citet{Milanic-sim} that $\si(K_{n,m})=\min\{n,m\}$. The argument used is that, in a $d$-simultaneous interval representation of a complete bipartite graph with bipartition $(X,Y)$, at most one of the sets of intervals representing $X$ and $Y$ can contain two disjoint intervals, so the other set consists of pairwise intersecting intervals that must have pairwise disjoint label sets, thus $d$ is at least the size of that set. We will make here a further observation. 

\begin{lemma}
\label[lemma]{prop:si-interval-module}
    Let $H$ be a module of a graph $G$, $|H| \geq 2$, that induces an interval graph that is not complete. Let $G'$ be the graph obtained from $G$ by deleting all but two nonadjacent vertices in $H$, namely $v$ and $w$. Then the minimum $d$ such that $G$ admits a $d$-simultaneous interval representation such that there are two disjoint intervals representing vertices of $H$ is the same as the minimum $d$ such that $G'$ admits a $d$-simultaneous interval representation such that the intervals representing $v$ and $w$ are disjoint.
\end{lemma}
\begin{proof}
    Suppose first we have a $d$-simultaneous interval representation of $G$ such that there are two disjoint intervals representing vertices of $H$. Since the intervals are disjoint, the two vertices represented by them are not neighbors. So, deleting all the intervals representing vertices of $H$ except for two disjoint intervals gives a $d$-simultaneous interval representation of $G'$ such that the intervals representing $v$ and $w$ are disjoint.
    
    Suppose now that we have a $d$-simultaneous interval representation of $G'$ such that the intervals representing $v$ and $w$ are disjoint, and suppose without loss of generality that the interval $A$ representing $v$ lies entirely at the left of the interval $B$ representing $w$. Since $H$ is a module, any vertex that is adjacent to a vertex in $H$ must be adjacent to both $v$ and $w$, so its corresponding interval intersects both the right endpoint of $A$ and the left endpoint of $B$. Therefore, there is a sub-interval of $A$, sharing its right endpoint, that is contained in all the intervals corresponding to vertices that are adjacent to the vertices of $H$. We can then replace $A$ with a subdivision of that sub-interval into an interval representation for $H$, with all the vertices of $H$ having the same label set as $A$. As $H$ is not complete, there exist two non-adjacent vertices, and the intervals for these vertices in an interval representation of $H$ cannot intersect. Thus, we obtain a $d$-simultaneous interval representation of $G$ such that there are two disjoint intervals representing vertices of $H$.
\end{proof}

As a corollary of \cref{prop:si-contract-complete-module}, we have the following result. 
\begin{lemma}
    \label[lemma]{lemma:si-tc}
    For every graph $G$ with twin-cover number $k$, $\si(G) \leq \max\{1,k(k-1)/2 + k\}$.
\end{lemma}
\begin{proof}
    Let $X$ be a twin-cover of $G$ of size $k$. The connected components of $G\setminus X$ are cliques of vertices that are twins. We represent the {different} cliques by pairwise disjoint intervals, and all the vertices of a same clique by identical intervals, in such a way that intervals of cliques having no neighbors in $X$, {if any,} are all to the left of intervals of cliques having neighbors in $X$. We represent the vertices of $X$ by identical intervals intersecting all the intervals of cliques having neighbors in $X$. We use a set of at most $k(k-1)/2$ labels to correctly represent the adjacencies of $G[X]$, by assigning a label to each edge of $G[X]$ and giving to each vertex the set of labels of its incident edges in $G[X]$. We further give to each of the vertices of $X$ an extra label (all of them different), and we assign to each vertex in $V(G)\setminus X$ the set of extra labels corresponding to its neighbors in $X$. This makes also pairwise adjacent the vertices of each clique having neighbors in $X$. Finally, we
    use an arbitrary label for all the remaining vertices, in order to make pairwise adjacent the vertices of each clique having no neighbors in $X$.
\end{proof}

The twin-cover number can be arbitrarily larger than the simultaneous interval number. Indeed, paths have simultaneous interval number~1 and arbitrarily large twin-cover number.

We can use \cref{prop:si-contract-complete-module,prop:si-module-pairwise-intersecting,prop:si-interval-module} to prove the following. 
\begin{theorem}
    \label{theorem:si-param-cm}
    {\sc Simultaneous Interval Number} is \FPT when parameterized by the $\cluster$-modular cardinality of the input graph plus the solution size {(the input parameter $d$)}.
\end{theorem}
\begin{proof}
    Let $G$ be a graph with $\cluster$-modular cardinality $k$. We want to decide whether it admits a $d$-simultaneous interval representation. If $k=1$, then $\si(G) = 0$ or $1$ depending on the presence of edges. So, suppose $k > 1$. We can assume that there are no isolated cliques in $G$, because in that case the simultaneous interval number of $G$ will be the same as the simultaneous interval number of the graph obtained from $G$ by removing the isolated cliques. We can also assume that $d \geq 2$, since graphs with $\si(G) = 0$ are exactly the edgeless graphs, and graphs with $\si(G) \leq 1$ are exactly the interval graphs~\citep{Milanic-sim}, that can be recognized in polynomial time~\citep{Habib-interval-rec}. 
    
    Let $M_1, \dots, M_k$ be the modules in an optimal $\cluster$-modular partition, which can be computed in linear time  by \cref{prop:calculate-cluster-module-partition}. The cliques in each cluster $M_i$ are in turn modules of $G$, so we can apply \cref{prop:si-contract-complete-module} and contract each of them into a vertex, obtaining a graph $G'$ and modules $M'_1, \dots, M'_k$ which are independent sets, such that $\si(G)=\si(G')$ and a $d$-simultaneous interval representation of $G$ can be obtained from a $d$-simultaneous interval representation of $G'$. 
    
    In such a representation of $G'$, for each $M'_i$, either there exist two disjoint intervals, or the intervals are pairwise intersecting. We will then solve at most $2^k$ constrained instances, where we specify for each $1 \leq i \leq k$ such that $|M'_i|>1$, whether in the representation of $M'_i$ there exist two disjoint intervals, or the intervals are pairwise intersecting. In the first case, by \cref{prop:si-interval-module}, we can keep just 2 vertices of $M'_i$. In the second case, at least $|M'_i|$ labels are needed, since the intervals corresponding to vertices of $M'_i$ require pairwise disjoint label sets. So, if $|M'_i| > d$, we can discard the instance. If $|M'_i| \leq d$, by \cref{prop:si-module-pairwise-intersecting}, we may restrict ourselves to solutions in which all vertices of $M'_i$ are represented by the same interval. This preprocessing of each instance can be done in polynomial time.    

    For each constrained instance, we need to find a $d$-simultaneous interval representation of $G'$ with a fixed number $t$ of intervals, with $k \leq t \leq 2k$, where some pairs of them are meant to be disjoint and represent one vertex each, and some others represent a fixed number (at most $d$) of pairwise non-adjacent vertices. For a representation of $t$ intervals we can consider without loss of generality $2t$ linearly ordered endpoints, from which we choose $4$ endpoints for each of the $t-k$ pairs of disjoint intervals and $2$ endpoints for each of the $2k-t$ remaining intervals, and we discard the clearly unfeasible representations (those in which intervals that should intersect do not). A very loose upper bound for the number of representations is $(2t)!$, the number of permutations of all the interval endpoints, which is at most $(4k)!$ because $t \leq 2k$. As for the labels, we can try all the assignments of subsets of $\{1, \dots, d\}$ to the intervals, including the (at most $d$) copies of the $2k-t$ intervals that possibly represent more than one vertex. Since $d \geq 2$, to each of the $k$ modules correspond at most $d$ intervals, so the number of possible assignments is at most $(2^d)^{dk} = 2^{d^2k}$.
    
    Summarizing, we analyze $\O(2^k \cdot (4k)! \cdot 2^{d^2k})$ candidate solutions, and checking each of them can be done in time $\O(d(dk)^2)$, which proves that {\sc Simultaneous Interval Number} is \FPT when parameterized by the cluster-modular cardinality of the input graph plus the solution size.
\end{proof}

As a corollary of \cref{theorem:si-param-cm} and \cref{lemma:si-tc}, we have the following result.

\begin{corollary}
    \label{cor:si-param-twin-cover}
    {\sc Simultaneous Interval Number} is \FPT when parameterized by the size of a minimum twin-cover of the input graph.
\end{corollary}
\begin{proof}
    Let $G$ be a graph with twin-cover number $k$. By \cref{lemma:si-tc}, $\si(G) \leq \max\{1,k(k-1)/2 + k\}$, and by \cref{lemma:twin-cover-hcupn}, $\cm(G) \leq 2^k + k$. We thus use the algorithm from \cref{theorem:si-param-cm} to obtain an \FPT algorithm parameterized by $k$.
\end{proof}

As corollaries of \cref{theorem:si-param-cm} and \cref{cor:si-param-twin-cover}, and the relations in \cref{fig:diagram}, we have also the following results. 
\begin{corollary}
    \label{corol:si-param-nd-vc}
    {\sc Simultaneous Interval Number} is \FPT when parameterized by the neighborhood diversity of the input graph plus the solution size $d$, and by the size of a minimum vertex cover of the input graph.
\end{corollary}

\section{Conclusions and open problems}
\label{sec:open-problems}
We summarize our complexity and algorithmic results in \cref{tab:results}.  

\begin{table}[h]
\footnotesize
\centering
\begin{tabular}[t]{lcccccc}
\toprule
Problem & General & $\intervalmc$ & $\cm$ & $\nd$ & $\tc$ & $\vc$ \\
\midrule
\textsc{$\interval$-MC}  & linear & & & & &  \\
\textsc{$\cluster$-MC}  & linear & & & & & \\
\textsc{Thinness}  & ~\NP-c~ & linear kernel & linear kernel & linear kernel & ~\FPT~ & \FPT \\
\textsc{Sim. Int. Num.}~~~  & ~\NP-c~ & ? & \FPT $(\cm + d)$ & \FPT $(\nd + d)$ & ~\FPT~ & \FPT\\
\bottomrule
\end{tabular}
\caption{Parameterized complexity results. The parameter $d$ represents the solution size.}\label{tab:results}
\end{table}%

Possible avenues of research that this work leaves open include:

\begin{itemize}
    \item Is \textsc{Thinness} \FPT (or even \textsf{XP}) when parameterized by the solution size, that is, the desired thinness of the input graph? Note that both a positive or a negative answer are compatible with the non-existence of a polynomial kernel proved in \cref{theorem:thinness-kernelization-lower-bound}. Combining the results by \citet{tesis-diego} and \citet{thinness-np-complete}, we can reduce an instance of the \textsc{Non-Betweenness} problem to an instance of \textsc{Thinness}. The thinness of the constructed graph depends in part on the size of the \textsc{Non-Betweenness} instance, and thus leaves open the possibility of a polynomial-time algorithm for \textsc{Thinness} when the thinness of the input graph is bounded by a constant.
    \item Is \textsc{Thinness} \FPT (or even \textsf{XP}) when parameterized by other parameters, like for example the treewidth of the input graph? The total order involved in the definition of thinness makes it unclear to be amenable to a {\sf CMSO} logic formulation in order to apply Courcelle's theorem~\citep{Courcelle90}. On the other hand, note that \cref{theorem:thinness-kernelization-lower-bound} implies that \textsc{Thinness} is unlikely to admit polynomial kernels parameterized by treewidth. 
    \item Does \textsc{Thinness} admit polynomial kernels when parameterized by the vertex cover or the twin-cover number of the input graph?
    \item {Is \textsc{Simultaneous Interval Number}  \FPT (or even \textsf{XP}) when parameterized by the $\cluster$-modular cardinality or the $\interval$-modular cardinality?}
\end{itemize}

\bibliographystyle{abbrvnat}
\bibliography{references,bnm-j,bnm}
\label{sec:biblio}

\newpage
\begin{appendices}
\crefalias{section}{appendix} 

\section{Modular-width vs. linear mim-width}\label{sec:modw-vs-lmimw}

Let $G$ be a graph. Given a linear layout $v_1,v_2, \dots, v_n$ of $V(G)$, the \textit{mim of a cut} $v_1, \dots, v_i | v_{i+1}, \dots, v_n$
is the maximum number of edges of an induced matching in the bipartite graph with bipartition $A_i\coloneqq\{v_1, \dots, v_i\}, B_i\coloneqq\{v_{n+1}, \dots, v_n\}$ and edge set formed by the edges o $G$ with one endpoint in $A_i$ and the other in $B_i$. The \textit{mim-width of the linear layout} is the maximum mim over its cuts, and the \textit{linear mim-width of graph $G$}, denoted by $\lmimw(G)$, is the minimum mim-width over its linear layouts.
\citet{Saether-branch} proved the \NP-hardness for deciding both mim-width and linear mim-width, and, moreover, that they do not admit either \FPT or constant-factor approximation algorithms under common complexity assumptions.

We describe here a graph family of bounded modular-width and unbounded linear mim-width. 

In a graph $G$, the \emph{substitution} operation of a vertex $v \in V(G)$ by a graph $H$
results in a graph $G'$ such that $V(G') = V(G) \cup V(H) \setminus
\{v\}$, and  vertices $x$, $y$ are adjacent in $G'$ if either $x, y \in V(G)
\setminus \{v\}$ and $xy \in E(G)$, or $x, y \in V(H)$ and $xy \in
E(H)$, or $x \in V(G) \setminus \{v\}$, $y \in V(H)$, and $xv \in
E(G)$.

A \emph{bipartite claw} is obtained from the claw $K_{1,3}$ by subdividing once each of the three edges. It is depicted as $H_1$ in \autoref{fig:bip-claw}.

Let us define the following family of graphs: the graph $H_0$ is the one-vertex graph;
for every $n \geq 0$, $H_{n+1}$ is obtained from a bipartite claw by substituting each of the leaves by a copy of the graph $H_n$ (see \autoref{fig:bip-claw}). 

\begin{figure}[ht]
    \centering
        \resizebox{.6\textwidth}{!}{%
\begin{tikzpicture}[scale=.7]
	\begin{pgfonlayer}{nodelayer}
		\node [style=nodo] (0) at (0, 0) {};
		\node [style=nodo] (1) at (-0.75, 0.75) {};
		\node [style=nodo] (2) at (0, -1) {};
		\node [style=nodo] (3) at (0.75, 0.75) {};
		\node [style=nodo] (4) at (-1.5, 1.5) {};
		\node [style=nodo] (5) at (0, -2) {};
		\node [style=nodo] (6) at (1.5, 1.5) {};
		\node [style=nodo] (7) at (-7, 0) {};
		\node [style=nodo] (8) at (7, 0) {};
		\node [style=nodo] (9) at (6.25, 0.75) {};
		\node [style=nodo] (10) at (7, -1) {};
		\node [style=nodo] (11) at (7.75, 0.75) {};
		\node [style=nodo] (12) at (5.5, 1.5) {};
		\node [style=nodo] (13) at (7, -2) {};
		\node [style=nodo] (15) at (5.5, 2.5) {};
		\node [style=nodo] (16) at (5, 3) {};
		\node [style=nodo] (17) at (5, 2) {};
		\node [style=nodo] (18) at (4.5, 2.5) {};
		\node [style=nodo] (19) at (4.5, 1.5) {};
		\node [style=nodo] (20) at (4, 2) {};
		\node [style=nodo] (29) at (6.25, -2.75) {};
		\node [style=nodo] (30) at (7.75, -2.75) {};
		\node [style=nodo] (31) at (6.25, -3.5) {};
		\node [style=nodo] (32) at (7, -2.75) {};
		\node [style=nodo] (33) at (7.75, -3.5) {};
		\node [style=nodo] (34) at (7, -3.5) {};
		\node [style=none] (38) at (-7, -4.5) {$H_0$};
		\node [style=none] (39) at (0, -4.5) {$H_1$};
		\node [style=none] (40) at (7, -4.5) {$H_2$};
		\node [style=nodo] (51) at (8.5, 1.5) {};
		\node [style=nodo] (52) at (9.5, 1.5) {};
		\node [style=nodo] (53) at (10, 2) {};
		\node [style=nodo] (54) at (9, 2) {};
		\node [style=nodo] (55) at (9.5, 2.5) {};
		\node [style=nodo] (56) at (8.5, 2.5) {};
		\node [style=nodo] (57) at (9, 3) {};
  \end{pgfonlayer}
	\begin{pgfonlayer}{edgelayer}
		\draw [style=arista] (0) to (3);
		\draw [style=arista] (3) to (6);
		\draw [style=arista] (0) to (2);
		\draw [style=arista] (2) to (5);
		\draw [style=arista] (0) to (1);
		\draw [style=arista] (1) to (4);
		\draw [style=arista] (8) to (11);
		\draw [style=arista] (8) to (10);
		\draw [style=arista] (10) to (13);
		\draw [style=arista,bend right=10] (9) to (15);
		\draw [style=arista,bend right=40] (9) to (16);
		\draw [style=arista,bend right=20] (9) to (17);
		\draw [style=arista,bend right=30] (9) to (18);
		\draw [style=arista,bend left=10] (9) to (19);
		\draw [style=arista,bend left=40] (9) to (20);
		\draw [style=arista] (8) to (9);
		\draw [style=arista] (9) to (12);  
		\draw [style=arista] (12) to (15);
		\draw [style=arista] (15) to (16);
		\draw [style=arista] (12) to (19);
		\draw [style=arista] (12) to (17);
		\draw [style=arista] (17) to (18);
		\draw [style=arista] (19) to (20);
		\draw [style=arista,bend right=10] (10) to (29);  
		\draw [style=arista,bend right=40] (10) to (31);  
		\draw [style=arista] (13) to (29);
		\draw [style=arista] (29) to (31);
		\draw [style=arista,bend right=20] (10) to (32);  
		\draw [style=arista,bend left=10] (10) to (30);  
		\draw [style=arista,bend left=40] (10) to (33);  
		\draw [style=arista] (13) to (32);
		\draw [style=arista] (13) to (30);
		\draw [style=arista] (30) to (33);
		\draw [style=arista] (32) to (34);
		\draw [style=arista,bend right=30] (10) to (34);  
            \draw [style=arista,bend right=10] (11) to (52);  
		\draw [style=arista,bend right=40] (11) to (53);  
		\draw [style=arista] (51) to (52);
		\draw [style=arista] (52) to (53);
		\draw [style=arista] (51) to (56);
		\draw [style=arista,bend left=10] (11) to (56);  
		\draw [style=arista,bend left=40] (11) to (57);  
		\draw [style=arista] (51) to (54);
		\draw [style=arista,bend right=20] (11) to (54);  
		\draw [style=arista,bend right=30] (11) to (55);  
		\draw [style=arista] (54) to (55);
		\draw [style=arista] (56) to (57);
		\draw (11) to (51);
	\end{pgfonlayer}
\end{tikzpicture}}
    \caption{The graphs $H_0$, $H_1$, and $H_2$.}
    \label{fig:bip-claw}
\end{figure}
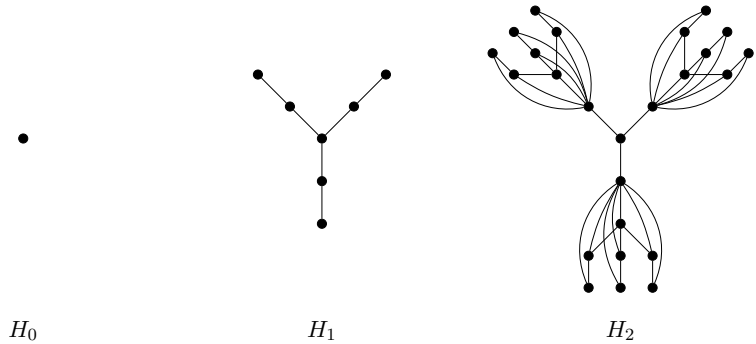

The graphs in $\{H_n\}_{n\geq 0}$ have modular-width~7, since in the modular decomposition tree, every prime node is a bipartite claw, a 7-vertex graph. We will show next that $\lmimw(H_n)=n$ for every $n \geq 0$. The construction and the proof are inspired by the construction by \citet{H-A-R-lin-mim-trees} of a family of trees with unbounded linear mim-width.

\begin{figure}[ht]
    \centering
        \resizebox{\textwidth}{!}{%
\begin{tikzpicture}[scale=0.8]

\begin{scope}[xshift=0cm,yshift=0cm]

\minivertex{5}{4}{w0};
\minivertex{5.5}{4}{w1};
\minivertex{6}{4}{w2};
\minivertex{6.5}{4}{w3};
\minivertex{7}{4}{w4};
\minivertex{7.5}{4}{w5};
\minivertex{8}{4}{w6};

\path (w0) edge [e1] (w1);
\path (w2) edge [e1] (w3);
\path (w4) edge [e1] (w6);
\path (w1) edge [e1,bend left=45] (w4);
\path (w2) edge [e1,bend left=30] (w4);

\end{scope}

\begin{scope}[xshift=7.7cm,yshift=-2cm]

\minivertex{5}{4}{w0};
\minivertex{5.5}{4}{w1};
\minivertex{6}{4}{w2};
\minivertex{6.5}{4}{w3};
\minivertex{7}{4}{w4};
\minivertex{7.5}{4}{w5};
\minivertex{8}{4}{w6};

\path (w0) edge [e1] (w1);
\path (w2) edge [e1] (w3);
\path (w4) edge [e1] (w6);
\path (w1) edge [e1,bend left=45] (w4);
\path (w2) edge [e1,bend left=30] (w4);

\end{scope}

\begin{scope}[xshift=-7.7cm,yshift=-2cm]

\minivertex{5}{4}{w0};
\minivertex{5.5}{4}{w1};
\minivertex{6}{4}{w2};
\minivertex{6.5}{4}{w3};
\minivertex{7}{4}{w4};
\minivertex{7.5}{4}{w5};
\minivertex{8}{4}{w6};

\path (w0) edge [e1] (w1);
\path (w2) edge [e1] (w3);
\path (w4) edge [e1] (w6);
\path (w1) edge [e1,bend left=45] (w4);
\path (w2) edge [e1,bend left=30] (w4);

\end{scope}

\begin{scope}[xshift=0cm,yshift=-2cm]

\minivertex{5}{4}{w0};
\minivertex{5.5}{4}{w1};
\minivertex{6}{4}{w2};
\minivertex{6.5}{4}{w3};
\minivertex{7}{4}{w4};
\minivertex{7.5}{4}{w5};
\minivertex{8}{4}{w6};

\path (w0) edge [e1] (w1);
\path (w2) edge [e1] (w3);
\path (w4) edge [e1] (w6);
\path (w1) edge [e1,bend left=45] (w4);
\path (w2) edge [e1,bend left=30] (w4);

\end{scope}

\begin{scope}[xshift=0cm,yshift=0cm]

\vertex{1.3}{2}{v1};
\vertex{2.3}{2}{v2};
\vertex{9.3}{2}{v3};
\vertex{10.3}{2}{v4};

\vertex{4.3}{1}{v5};
\vertex{4.8}{1}{v6};
\vertex{5.3}{1}{v7};
\vertex{5.8}{1}{v8};
\vertex{6.3}{1}{v9};
\vertex{6.8}{1}{v10};
\vertex{7.3}{1}{v11};

\vertex{4.3+7.7}{1}{w5};
\vertex{4.8+7.7}{1}{w6};
\vertex{5.3+7.7}{1}{w7};
\vertex{5.8+7.7}{1}{w8};
\vertex{6.3+7.7}{1}{w9};
\vertex{6.8+7.7}{1}{w10};
\vertex{7.3+7.7}{1}{w11};

\vertex{4.3-7.7}{1}{z5};
\vertex{4.8-7.7}{1}{z6};
\vertex{5.3-7.7}{1}{z7};
\vertex{5.8-7.7}{1}{z8};
\vertex{6.3-7.7}{1}{z9};
\vertex{6.8-7.7}{1}{z10};
\vertex{7.3-7.7}{1}{z11};

\node[draw, gray!30!black,fill=gray!50,rounded corners=3,opacity=0.5, fit={(3.3,3.6) (8.3,4.5)}, inner sep=0.25pt, label=center:{\phantom{j}{\small Prime}\phantom{jaaaaaaaaaaaaaa}}] (F) {};
\node[draw, gray!30!black,fill=gray!50,rounded corners=3,opacity=0.5, fit={(3.3,1.6) (8.3,2.5)}, inner sep=0.25pt, label=center:{\phantom{j}{\small Prime}\phantom{jaaaaaaaaaaaaaa}}] (G) {};
\node[draw, gray!30!black,fill=gray!50,rounded corners=3,opacity=0.5, fit={(-4.4,1.6) (0.6,2.5)}, inner sep=0.25pt, label=center:{\phantom{j}{\small Prime}\phantom{jaaaaaaaaaaaaaa}}] (H) {};
\node[draw, gray!30!black,fill=gray!50,rounded corners=3,opacity=0.5, fit={(11,1.6) (16,2.5)}, inner sep=0.25pt, label=center:{\phantom{j}{\small Prime}\phantom{jaaaaaaaaaaaaaa}}] (I) {};
\draw (F)--(v1);
\draw (F)--(v2);
\draw (F)--(v3);
\draw (F)--(v4);
\draw (F)--(G);
\draw (F)--(H);
\draw (F)--(I);
\draw (G)--(v5);
\draw (G)--(v6);
\draw (G)--(v7);
\draw (G)--(v8);
\draw (G)--(v9);
\draw (G)--(v10);
\draw (G)--(v11);
\draw (I)--(w5);
\draw (I)--(w6);
\draw (I)--(w7);
\draw (I)--(w8);
\draw (I)--(w9);
\draw (I)--(w10);
\draw (I)--(w11);
\draw (H)--(z5);
\draw (H)--(z6);
\draw (H)--(z7);
\draw (H)--(z8);
\draw (H)--(z9);
\draw (H)--(z10);
\draw (H)--(z11);

\end{scope}
\end{tikzpicture} 
}%
\caption{Modular decomposition of the graph $H_2$.}
    \label{fig:H2-decomp}
\end{figure}
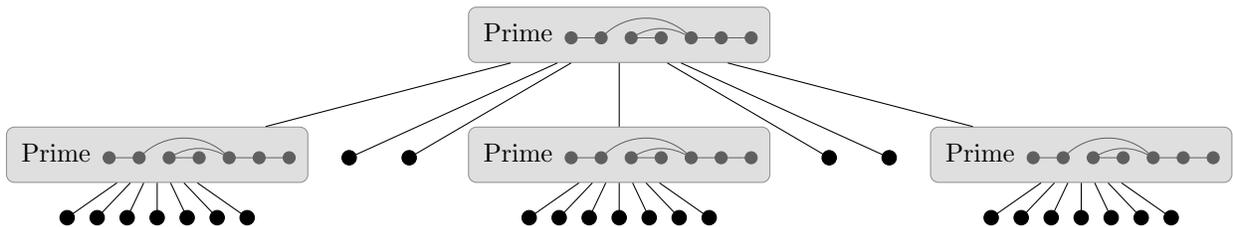

\begin{theorem}\label{incomp}
The family of graphs $\{H_n\}_{n\geq 0}$ has unbounded linear mim-width. More precisely, for every $n \geq 0$ we have that $\lmimw(H_n)=n$. In particular, there are classes of graphs of bounded modular-width and unbounded linear mim-width.
\end{theorem}

\begin{proof} 
We prove this by induction. Since $H_0$ is the one-vertex graph, $\lmimw(H_0)=0$. Since $H_1$ has at least one edge, $\lmimw(H_1)\geq 1$. On the other hand, the linear layout of $H_1$ (the bipartite claw) shown inside each prime node in \autoref{fig:H2-decomp} has mim-width~$1$, thus $\lmimw(H_1)=1$. Suppose now $n \geq 1$, and consider $H_{n+1}$.  Recall that 
$H_{n+1}$ is obtained from a bipartite claw by substituting each of the leaves by a copy of the graph $H_n$. Let $H_n^1$, $H_n^2$, and $H_n^3$ be the three copies of $H_n$, let $x$ be the center of the claw, and let $u_1, u_2, u_3$ be the vertices such that $N(u_i) = V(H_n^i) \cup \{x\}$ for $i \in \{1,2,3\}$. By inductive hypothesis, $\lmimw(H_n)=n$, and since $H_n$ is an induced subgraph of $H_{n+1}$ and the parameter is monotone with respect to induced subgraphs, $\lmimw(H_{n+1}) \geq n$. 

We prove that $\lmimw(H_{n+1}) \geq n+1$ by contradiction. Suppose $\lmimw(H_{n+1}) = n$, so there exists a linear layout of $\lmimw(H_{n+1})$ of mim-width $n$. We know that for each $i \in \{1,2,3\}$ we have a cut $C_i$ of mim $n$ where all the $n$ edges of the induced matching are edges of $H_n^i$. By symmetry, without loss of generality we assume these three cuts come in the order $C_1,C_2,C_3$. Note that, in the layout, all the vertices of $H_n^1$ must appear before $C_2$ and all the vertices of $H_n^3$ after $C_2$, as otherwise, since $H_n^1$ and $H_n^3$ are connected and the distance between $H_n^2$ and each of the other two copies of $H_n$ is $4$, there would be an extra edge crossing $C_2$ that would increase the mim of this cut to $n + 1$. Similarly, $u_1$ has to be placed before $C_2$ and $u_3$ has to be placed after $C_2$, since $u_i$ is connected to $H_n^i$ and both $u_1$ and $u_3$ are at distance $3$ from $H_n^2$. But then the vertex $x$ cannot be placed before $C_2$ or after $C_2$ without increasing the mim of this cut by adding at least one of the edges $u_1x$ or $u_3x$ to the induced matching. By contradiction, we conclude that $\lmimw(H_{n+1}) \geq n+1$.

A linear layout of $H_{n+1}$ with mim-width $n+1$ can be obtained from the linear layout of the bipartite claw shown inside each prime node in \autoref{fig:H2-decomp} (that has mim-width~$1$), by replacing each leaf with a linear layout of $H_{n}$ with mim-width $n$. That is, an optimal linear layout of $H_n^1$, $u_1$, $u_2$, an optimal linear layout of $H_n^2$, $x$, $u_3$, an optimal linear layout of  $H_n^3$.
\end{proof}

We leave as an open question whether there exist families with smaller modular-width and unbounded linear mim-width.

\section{Counterexample to reducing 2-thin modules}
\label{sec:counterexample}

Here we prove that graph $G'$ in \cref{fig:counterexample} is 2-thin, while $G$ is not, showing that replacing a module of thinness 2 with any other graph of thinness 2 can modify the thinness of the graph.

\begin{figure}[h]
    \begin{subfigure}{0.49\linewidth}
        \centering
        \includegraphics[width=0.4\linewidth]{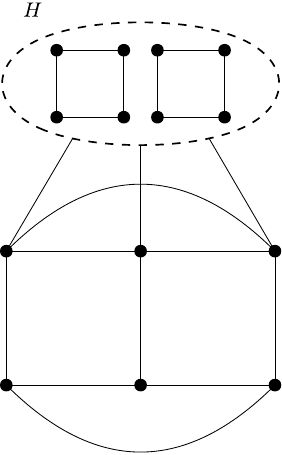}
        \caption{$\thin(G) = 3$.}
    \end{subfigure}
    \begin{subfigure}{0.49\linewidth}
        \centering
        \includegraphics[width=0.4\linewidth]{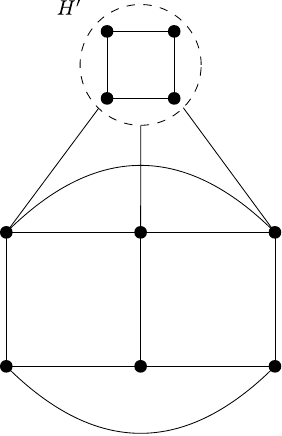}
        \caption{$\thin(G') = 2$.}
    \end{subfigure}
    \caption{For some graphs, replacing a module of thinness 2 with another graph of thinness 2 decreases the thinness of the graph. Here, dotted lines denote modules, meaning that every vertex adjacent to the dotted line is adjacent to every vertex inside the dotted line. Graph $G$ was found by running the algorithms by \citet{thinness-repo}.}
    \label{fig:counterexample}
\end{figure}

\begin{lemma}
\label[lemma]{prop:G'-has-thinness-2}
    Graph $G'$ in \cref{fig:counterexample} has thinness 2.
\end{lemma}
\begin{proof}
    \cref{fig:solution} shows a consistent solution for graph $G'$ using two classes. Note that $G'$ contains a cycle on 4 vertices as an induced subgraph, which is not an interval graph. Thus, we have that $\thin(G') > 1$, since the thinness is a hereditary property.
    \begin{figure}[h]
        \centering
        \includegraphics[width=0.25\linewidth]{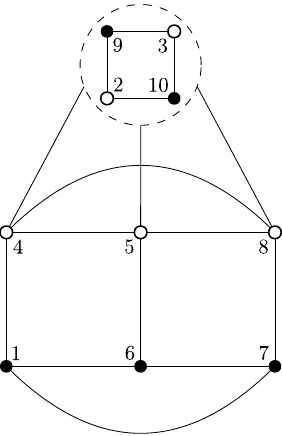}
        \caption{A consistent solution using two classes for graph $G'$ of \cref{fig:counterexample}. The numbers denote the order of the vertices. The filled-in vertices form one class of the partition, and the hollowed-out vertices the other one.}
        \label{fig:solution}
    \end{figure}
\end{proof}

The rest of this section is devoted to proving the following lemma.

\begin{lemma}
\label[lemma]{prop:G-not-2-thin}
    Graph $G$ in \cref{fig:counterexample} is not 2-thin.
\end{lemma}

In the consistent solution shown in \cref{fig:solution}, the vertices in $H'$ are not consecutive in the order. This is not a coincidence: there is in fact no optimal consistent solution for $G'$ where the vertices in $H'$ appear consecutively. This is crucial for establishing \cref{prop:G-not-2-thin}. These statements will be proved in \cref{prop:no-consecutive-solution,prop:replacing-module-maintains-thinness-if-irreducible-clique-consecutive}.

We begin with a standard definition used when analyzing the thinness of a graph. 

\begin{definition}[Incompatibility graph \citep{tesis-diego}]
\label{def:incompatibility-graph}
    The \emph{incompatibility graph} $G_<$ of a graph $G$ with order $<$ is a graph such that 
    \begin{itemize}
        \item $V(G_<) = V(G)$, and
        \item for every triple of vertices $u < v < w$, if $(u, w) \in E(G)$ and $(v, w) \not\in E(G)$, then $(u, v) \in E(G_<)$.
    \end{itemize}
\end{definition}

The definition implies that in the incompatibility graph $G_<$ of a graph $G$ with order $<$, the two endpoints of each edge cannot belong to the same class in a partition of $V(G)$ consistent with $<$. 


\citet{thinness-of-product-graphs} showed that in a consistent solution for a thinness $k$ graph there is always a set of $k$ vertices that belong to different classes such that each vertex in the set has a non-neighbor greater than it. We formalize this with the following definition and \cref{prop:irreducible-clique}.
\begin{definition}[Irreducible clique]
    An \emph{irreducible clique of size $\ell$} of an incompatibility graph $G_<$ of a graph $G$ is a clique $C \coloneqq\{v_1, \dots,v_\ell\}$ of $G_<$ such that for every vertex $v_i \in C$ there exists another vertex $w > v_i$ such that $(v_i, w) \not\in E(G)$.
\end{definition}

\begin{lemma}[{\citep[Lemma 26]{thinness-of-product-graphs}}]
\label[lemma]{prop:irreducible-clique}
    Let $G$ be a graph of thinness $k$. For every order $<$ of the vertices of $G$ there exists an irreducible clique of size $k$ in $G_<$.
\end{lemma}

We now prove some properties of the order between the vertices inside and outside a module in a graph.

\begin{lemma}
\label[lemma]{prop:vertex-that-breaks-irreducible-clique-is-adjacent}
    Let $G$ be a graph with a module $M$. Let $(<, S)$ be a consistent solution for $G$. If there exist three vertices $u < v < w$ such that
    \begin{itemize}
        \item $u$ and $w$ belong to an irreducible clique of the incompatibility graph $G[M]_<$,
        \item $v\in V(G)\setminus M$, and
        \item $v$ belongs to the same class as $u$ in $S$;
    \end{itemize}
    then $v$ is adjacent to the vertices in $M$.
\end{lemma}
\begin{proof}
    Otherwise, there would be a vertex $x \in M$ adjacent to $u$ greater than $w$ (and thus greater than $v$) that is not adjacent to $v$, which would make the triple $\{u, v, x\}$ inconsistent.
\end{proof}

\begin{lemma}
\label[lemma]{prop:neighbors-of-modules-that-share-class-are-greater}
    Let $G$ be a graph, and $M$ be a module of $G$ with a neighbor $v$. In every consistent solution $(<, S)$ for $G$, if $v$ shares a class in $S$ with a vertex $u$ in an irreducible clique of the incompatibility graph $G[M]_<$, then $u < v$.
\end{lemma}
\begin{proof}
    By \hyperref[def:incompatibility-graph]{the definition of incompatibility graph}, there exists a vertex $w \in M$ that is not adjacent to $u$ such that $u < w$. Thus, $v$ must be larger than $u$, as otherwise $\{v, u, w\}$ would form an inconsistent triple.
\end{proof}

\begin{definition}
    The result of \emph{replacing a module $M$ of a graph $G$ with another graph $H$} is a graph $G'$ where $V(G') \coloneqq (V(G) \setminus M) \cup V(H)$ and the edges of $G'$ are the union of the edges of $H$ and the edges of $G \setminus M$, making the neighbors of $M$ in $G$ also be neighbors of every vertex of $V(H)$ in $V(G')$.
\end{definition}

In some cases, replacing a module with a graph with the same thinness does not increase the thinness of the graph. This is the case when the module has thinness 1, as shown in \cref{prop:contract-interval-module}, but this can actually be generalized to modules with greater thinness, as shown in the following lemma.

\begin{lemma}
\label[lemma]{prop:replacing-module-maintains-thinness-if-irreducible-clique-consecutive}
    Let $\ell \in \N$. Let $G$ be a graph with a module $M$, and $G'$ be the result of replacing $M$ with a thinness $\ell$ graph $H$ in $G$. Consider a consistent solution $(<, S)$ in $k$ classes for $G$, and suppose there exists an irreducible clique $C$ of $G[M]_<$ of size $\ell$. If there are no vertices in $V(G) \setminus M$ between the vertices of $C$ according to $<$ that share a class in $S$ with the vertices in $C$, then there exists a consistent solution in $k$ classes for $G'$ where all the vertices in $V(H)$ are consecutive in the order.
\end{lemma}
\begin{proof}
    Suppose $(<, S)$ is a consistent solution for $G$ in $k$ classes where the $\ell$ vertices $C$ of an irreducible clique of $G[M]_<$ have no vertex of $V(G) \setminus M$ between them. Let $(<_H, S_H)$ be an optimal consistent solution for $H$.
    
    We construct a solution $(<', S')$ in $k$ classes as follows.
    To construct $<'$, we remove all vertices of $M$ from the order $<$, and replace the greatest vertex of $C$ with the vertices in $V(H)$, ordered according to $<_H$. The partition $S'$ is defined by identifying the $\ell$ classes of $S_H$ with the $\ell$ different classes to which the vertices of $C$ belong in the partition $S$.
    
    We claim this solution is consistent. To prove this, we analyze each possible triple of vertices $u <' v <' w$ to see that they form a consistent triple.
    \begin{itemize}
        \item \underline{$\{u, v, w\} \subseteq V(H)$ or $\{u, v, w\} \subseteq V(G') \setminus V(H)$:} This triple of vertices is consistent in either $(<, S)$ or $(<_H,S_H)$, which are induced solutions of $(<', S')$.
        \item \underline{$u \in V(H)$ and $v, w \not\in V(H)$:} There exists a vertex $x \in C$ that belongs to the same class as $u$ in $S$ and also satisfies $x < v < w$. Thus, if $w$ is adjacent to $u$ and $v$ belongs to the same class as $u$ in $S'$, then $w$ is also adjacent to $x$, and $v$ belongs to the same class as $x$ in $S$. Therefore, as $(<, S)$ is a consistent solution, the set $\{u,v,w\}$ must be a consistent triple.
        \item \underline{$u \not\in V(H)$, $v \in V(H)$, and $w \not\in V(H)$:} There exists a vertex $x \in C$ in the same class as $v$ in $S$. There are two possibilities:
        \begin{enumerate}
            \item We have $x < u$, in which case $u$ appears between two vertices of $C$ in $<$. By hypothesis, vertex $u$ does not share class with $v$, and so this triple is consistent.
            \item Otherwise, we have $x > u$. As $\{u, x, w\}$ does not form an inconsistent triple in $(<, S)$, the same thing happens for $\{u, v, w\}$ in $(<', S')$.
        \end{enumerate}
        \item \underline{$u, v \not\in V(H)$ and $w \in V(H)$:} There is a vertex $x \in C$ that is greater than both $u$ and $v$ according to $<$, and that forms a consistent triple with them. As vertex $w$ has the same neighborhood as $x$ in $G \setminus H$, the triple $\{u, v, w\}$ is also consistent.
        \item \underline{$\{u,v\} \subseteq V(H)$ and $w \not\in V(H)$:} Vertex $w$ is either adjacent to both $u$ and $v$, or to none of them, so this triple is consistent.
        \item\underline{$u \in V(H), v \not\in V(H)$, and $w \in V(H)$:} There are no vertices outside $H$ between two vertices of $H$, so this case does not happen.
        \item\underline{$u \not\in V(H)$ and $\{v, w\} \subseteq V(H)$:} If $u$ and $v$ belong to the same class in $S'$, there is also a vertex $x \in C$ in the same class as $u$. By hypothesis, we have $u < x$. As $x$ belongs to an irreducible clique of $G[M]_<$, there exists a vertex $y \in M$ that is not a neighbor of $x$ such that $x < y$. As $(<, S)$ is a consistent solution, this means that $u$ is not a neighbor of the vertices in the module $M$, and thus is also not a neighbor of the vertices of $H$ in $G'$. Therefore, this triple is consistent.
    \end{itemize}
\end{proof}

The converse of \cref{prop:replacing-module-maintains-thinness-if-irreducible-clique-consecutive} is also true when referring to the graphs of \cref{fig:counterexample}.

\newcommand{\singleCycleGraph}{G'}
\newcommand{\singleCycleModule}{H'}
\newcommand{\doubleCycleGraph}{G}
\newcommand{\doubleCycleModule}{H}
\begin{lemma}
\label[lemma]{prop:when-replacing-maintains-thinness}
    Let $\singleCycleGraph$ be a thinness two graph with a module $\singleCycleModule$ that induces a cycle on 4 vertices. Define $\doubleCycleGraph$ to be the result of replacing $\singleCycleModule$ with the union of two cycles on 4 vertices. The thinness of $\doubleCycleGraph$ is equal to 2 if and only if there exists an optimal consistent solution $(<', S')$ for $\singleCycleGraph$ where the vertices in $\singleCycleModule$ appear consecutively in $<'$.

\end{lemma}
\begin{proof}
Let $\doubleCycleModule$ be the union of two cycles on 4 vertices. Note that $\thin(\singleCycleGraph[\singleCycleModule]) = \thin(\doubleCycleModule) = 2$.

$\impliedby$) By \cref{prop:irreducible-clique}, the incompatibility graph $\singleCycleGraph[\singleCycleModule]_{<'}$ contains an irreducible clique of size 2. Therefore, this implication is deduced directly from \cref{prop:replacing-module-maintains-thinness-if-irreducible-clique-consecutive}.

$\implies)$ Suppose for a contradiction that $\thin(\doubleCycleGraph) = 2$ and $\singleCycleModule$ does not appear consecutively in any optimal consistent solution for $\singleCycleGraph$.

Let $\Gamma \coloneqq (<, S)$ be an optimal consistent solution for $\doubleCycleGraph$. Let $u$ be the first vertex according to $<$ that belongs to $V(\doubleCycleModule)$, and define $C$ to be the set of vertices of the connected component of $u$ in $\doubleCycleModule$, which induces one of the two cycles of $\doubleCycleModule$. Additionally, denote the other cycle of $\doubleCycleModule$ as $C'\coloneqq V(\doubleCycleModule) \setminus C$. 

Note that in every possible order of $C$, vertex $u$ forms an irreducible clique with some other vertex $w$ in the corresponding incompatibility graph. Consider the solution $\Gamma'$ induced by the vertices of $\doubleCycleGraph \setminus C'$ in $\Gamma$. This solution $\Gamma'$ is also a solution for $\singleCycleGraph$ (up to isomorphism), and hence, by hypothesis, the module $\singleCycleModule$ does not appear consecutively in $\Gamma'$. Thus, by \cref{prop:replacing-module-maintains-thinness-if-irreducible-clique-consecutive}, there exists a vertex $v \in V(\doubleCycleGraph) \setminus \doubleCycleModule$ such that $u < v < w$. Additionally, vertex $v$ shares a class with $u$ in $S$, and by \cref{prop:vertex-that-breaks-irreducible-clique-is-adjacent}, it is adjacent to the vertices in $\doubleCycleModule$.

The vertices in $C'$ also contain a vertex $u'$ in the same class as $u$ in $S$, as $C'$ induces a cycle, which is a thinness two graph. By \cref{prop:neighbors-of-modules-that-share-class-are-greater}, vertex $u'$ must be smaller than $v$. Note that $u <' u'$, as $u$ is the first vertex of $\doubleCycleModule$. Also, by \hyperref[def:incompatibility-graph]{the definition of incompatibility graph}, there exists a vertex $x \in C$ that is greater than $w$ and is adjacent to $u$. We therefore have that $\{u, u', x\}$ forms an inconsistent triple, which contradicts the fact that $\Gamma$ is a consistent solution.
\end{proof}

The following lemma along with \cref{prop:when-replacing-maintains-thinness} will finish the proof.
\begin{lemma}
\label[lemma]{prop:no-consecutive-solution}
    There is no optimal consistent solution for graph $G'$ in \cref{fig:counterexample} in which all vertices in module $H'$ appear consecutively.
\end{lemma}
\begin{proof}
Suppose there exists a consistent solution $(<, S)$ for $G'$ in which all vertices of $H'$ appear consecutively. By \cref{prop:G'-has-thinness-2}, this solution partitions $V(G)$ into two classes. On the other hand, by \cref{prop:irreducible-clique} and the fact that $\thin(G[H']) = 2$, there exists an irreducible clique $\{u, v\}$ of size 2 in $G[H']_<$, where $u < v$. These two vertices belong to two different classes in $S$.

We designate the neighbors of $H'$ in $G'$ as $\{a,b,c\}$, and the non-neighbors as $\{d,e,f\}$, as shown in \cref{fig:names}.
\begin{figure}[h]
    \centering
    \includegraphics[width=0.25\linewidth]{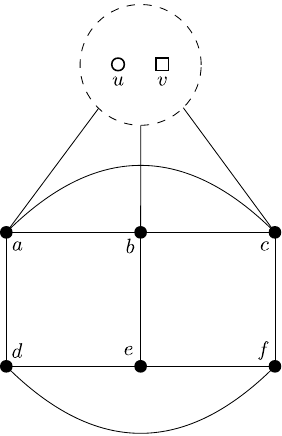}
    \caption{Names of the vertices of $G' \setminus H'$ in \cref{fig:counterexample}. Vertices $u$ and $v$ belong to different classes of $S$, indicated by the hollowed-out circle and square respectively.}
    \label{fig:names} 
\end{figure}

The vertices of $V(G') \setminus H'$ must belong to one of the two classes of $S$. Therefore, by \cref{prop:neighbors-of-modules-that-share-class-are-greater}, we have that $v < \{a,b,c\}$. Without loss of generality, assume that $a < b < c$.

\begin{claim}
\label{claim:non-neighbors-all-together}
    Either every vertex in $\{d,e,f\}$ is lower than $v$, or every vertex in $\{d,e,f\}$ is greater than $v$.
\end{claim}
\begin{proof}
Suppose this is not the case, and there exist two vertices $w_1, w_2 \in \{d,e,f\}$ such that $w_1 < v < w_2$. As $w_1$ and $w_2$ are adjacent but $v$ and $w_2$ are not, we have that $w_1$ and $v$ belong to different classes. By the same argument, if $u$ is also greater than $w_1$, they cannot belong to the same class either. Thus, $u < w_1$. However, by \hyperref[def:incompatibility-graph]{the definition of incompatibility graph}, there exists a vertex $x$ greater than $v$ in $H'$ that is adjacent to $u$. As $x$ is not adjacent to $w_1$, vertices $w_1$ and $u$ belong to different classes. As a result, $u$, $v$, and $w_1$ belong to three different classes, which contradicts the fact that $S$ consists of only two classes.
\end{proof}

Note that $e$ cannot be between $v$ and $a$, as otherwise $e$ and $a$ would form an inconsistent triple with $u$ or $v$, depending on the class to which $e$ belongs. In general, no vertex that is a non-neighbor of a vertex $y \in \{a,b,c\}$ can be between $v$ and $y$. 
This, combined with \cref{claim:non-neighbors-all-together}, means that either $\{d,e,f\} < v$ or $c < \{d,e\}$.

Suppose $\{d,e,f\} < v$. Rename $\{d,e,f\}$ as $\{w_1, w_2, w_3\}$ such that $w_1 < w_2 < w_3$. Vertex $w_1$ has a neighbor in $\{a,b,c\}$ which is not a neighbor of $w_2$ and $w_3$. Thus, $w_1$ does not belong to the same class as any of those two vertices. On the other hand, vertex $w_2$ also has a neighbor in $\{a,b,c\}$ not adjacent to $w_3$, and so it belongs to a different class from $w_2$. Hence, the vertices $\{w_1,w_2,w_3\}$ must all belong to different classes, while there are only two classes in $S$, arriving at a contradiction.

We therefore have that $c < \{d,e\}$. The same problem arises: vertex $a$ is adjacent to vertex $d$, which is not adjacent to $b$ and $c$; and vertex $b$ is adjacent to $e$, which is not adjacent to $c$. We thus reach a contradiction, and conclude that the statement of the lemma is satisfied.
\end{proof}

Finally, we arrive at a proof of \cref{prop:G-not-2-thin} by combining \cref{prop:no-consecutive-solution,prop:when-replacing-maintains-thinness,prop:G'-has-thinness-2}.

\section{Kernelization lower bounds}
\label{sec:lower-bounds}

In this section we present some parameterizations of \textsc{Thinness} {and of \textsc{Simultaneous Interval Number}} for which there are no polynomial kernels under some standard complexity assumptions. For this, we will use the cross-composition framework introduced by \citet{kernelization-cross-composition}. 
In their paper, they define the concept of \textsc{and}-cross-composition. First, they define the following relation. 

In this section, $\Sigma$ will denote a finite set of symbols, and $\Sigma^*$ will denote the set of finite strings of symbols in $\Sigma$. 

\begin{definition}[Polynomial equivalence relation~\citep{kernelization-cross-composition}]
\label{def:polynomial-equivalence-relation}
    An equivalence relation $\mathcal{R}$ on $\Sigma^*$ is called a \emph{polynomial equivalence relation} if the following two conditions hold:
    \begin{enumerate}
        \item There is an algorithm that given two strings $x, y \in \Sigma^*$ decides whether $x$ and $y$ belong to the same equivalence class in time polynomial in $\abs{x} + \abs{y}$.
        \item For any finite set $S \subseteq \Sigma^*$ the equivalence relation $\mathcal{R}$ partitions the elements of $S$ into a number of classes that is polynomially bounded by the size of the largest element of $S$.
    \end{enumerate}
\end{definition}

\begin{definition}[{\sc and}-cross-composition~\citep{kernelization-cross-composition}]
\label{def:and-cross-composition}
    Let $L \subseteq \Sigma^*$  be a language, let $\mathcal{R}$ be a polynomial equivalence relation on $\Sigma^*$, and let $\mathcal{Q} \subseteq \Sigma^* \times \N$ be a parameterized problem. An \emph{\textsc{and}-cross-composition of $L$ into $\mathcal{Q}$} (with respect to $\mathcal{R}$) is an algorithm that, given $t$ instances $x_1, x_2, \dots, x_t \in \Sigma^*$ of $L$ belonging to the same equivalence class of $\mathcal{R}$, takes time polynomial in $\sum_{i=1}^t \abs{x_i}$ and outputs an instance $(y, k) \in \Sigma^* \times \N$ such that $k$ is polynomially bounded in $\max_i \abs{x_i} + \log t$, and the instance $(y, k)$ belongs to $\mathcal{Q}$ if and only if all instances $x_i$ belong to $L$.
\end{definition}

The cross-composition framework is a generalization of earlier work by \citet{composition-kernelization}. A conjecture presented there was proved later by \citet{and-distillation-conjecture-proof}, which allows us to state the following theorem.

\begin{theorem}[\citep{kernelization-cross-composition}]
\label{theorem:and-cross-composition-kernelization-lower-bound}
    Let $L \subseteq \Sigma^*$ be an \NP-hard language and let $\mathcal{Q} \subseteq \Sigma^* \times \N$ be a parameterized problem. If $L$ \textsc{and}-cross-composes into $\mathcal{Q}$, then $\mathcal{Q}$ has no polynomial compression, assuming \Hip.
\end{theorem}

The main objective of this section is to prove the following theorem.
\begin{theorem}
\label{theorem:thinness-kernelization-lower-bound}
    Let $\p$ be a graph parameter such that for every graph $G$ we have $\p(G) \leq \max\{\abs{V(G)},\allowbreak \abs{E(G)}\}$, and such that if $H$ is the disjoint union of two graphs $G_1$ and $G_2$, we have $\p(H) \leq \max\{\p(G_1),\allowbreak \p(G_2)\}$. Then {\sc Thinness} {and {\sc Simultaneous Interval Number}} parameterized by $\p$ {have} no polynomial kernel assuming \Hip.
\end{theorem}

Examples of graph parameters that meet these criteria are treewidth, pathwidth, bandwidth, clique-width, mim-width, linear mim-width, modular-width, and even thinness {and simultaneous interval number}, as \cref{prop:thinness-union} shows.

Note that \cref{theorem:thinness-kernelization-lower-bound} does not exclude the possibility of a kernel, just of a \emph{polynomial} kernel. Whether \textsc{Thinness} admits any kernel parameterized by these parameters is still open.

We now prove \cref{theorem:thinness-kernelization-lower-bound}. The proof is quite simple, but we include it here for completeness.

\begin{proof}[of \cref{theorem:thinness-kernelization-lower-bound}]
    We define \textsc{and}-cross-compositions of \textsc{Thinness} into \textsc{Thinness($\p$)} {and of \textsc{Simultaneous Interval Number} into \textsc{Simultaneous Interval Number($\p$)}}, {where $\Pi(\p)$ denotes the problem $\Pi$ parameterized by $\p$}. As {both} \textsc{Thinness} {and \textsc{Simultaneous Interval Number} are} \textsf{NP}-hard, we can use \cref{theorem:and-cross-composition-kernelization-lower-bound} to prove that \textsc{Thinness($\p$)} {and \textsc{Simultaneous Interval Number($\p$)} have} no polynomial kernel, unless $\textsf{NP} \subseteq \textsf{coNP} / \textsf{poly}$.

    First, we define the relation $\mathcal{R}$ on instances of \textsc{Thinness} {(resp. \textsc{Simultaneous Interval Number})} such that two instances $(G_1, k_1)$ and $(G_2, k_2)$ are related if and only if $k_1 = k_2$. This relation fits in \cref{def:polynomial-equivalence-relation}, as checking whether $k_1 = k_2$ takes time polynomial in $\abs{(G_1, k_1)} + \abs{(G_2, k_2)}$, and every subset $S$ of instances of \textsc{Thinness} is partitioned into a number of classes not greater than the maximum number of vertices in the graph of an instance (as $\thin(G)\leq |V(G)|$, we can consider only the instances with $k \leq \abs{V(G)}$). {Respectively, every subset $S$ of instances of \textsc{Simultaneous Interval Number} is partitioned into a number of classes not greater than the maximum number of edges in the graph of an instance (as $\si(G)\leq |E(G)|$, we can consider only the instances with $k \leq \abs{E(G)}$).}
    
    Let $(G_1, k), \dots, (G_t, k)$ be $t$ instances of \textsc{Thinness} {(resp. \textsc{Simultaneous Interval Number})} in the same equivalence class of $\mathcal{R}$. 
    
    We construct an instance $((G', k), p)$ of $\textsc{Thinness}(\p)$ {(resp. \textsc{Simultaneous Interval Number}$(\p)$)} where $G' \coloneqq \bigcup_{i=1}^t G_i$ and $p \coloneqq \max_i \abs{V(G_i)}$ {(resp. $p \coloneqq \max_i \abs{E(G_i)}$)}. This can be done in time polynomial in $\sum_{i=1}^t \abs{(G_i, k)}$, and $p$ is polynomially bounded by $\max_i \abs{(G_i, k)} + \log t$. Also, $p$ is greater than or equal to $\p(G')$, as we required $\p(G')$ to be not greater than the maximum of $\p(G_i)$ for each $G_i$. Finally, by \cref{prop:thinness-union}, graph $G'$ is $k$-thin if and only if all graphs $G_i$ are $k$-thin, {and graph $G'$ admits a $d$-simultaneous interval representation if and only if all graphs $G_i$ admit one,}
    so $((G', k), p)$ belongs to $\textsc{Thinness}(\p)$ {(resp. $\textsc{Simultaneous Interval Number}(\p)$)} if and only if all instances $(G_1, k), \dots, (G_t, k)$ belong to \textsc{Thinness} {(resp. \textsc{Simultaneous Interval Number})}.

    We described \textsc{and}-cross-compositions of \textsc{Thinness} into \textsc{Thinness($\p$)} {and of \textsc{Simultaneous Interval Number} into \textsc{Simultaneous Interval Number($\p$)}}, and thus \textsc{Thinness($\p$)} {and \textsc{Simultaneous Interval Number($\p$)} have} no polynomial kernel unless \notHip.
\end{proof}
\end{appendices}

\end{document}